\documentclass{article}
\pdfoutput=1
\usepackage[margin=1in]{geometry}

\usepackage{tikz}
\usepackage{amsthm,amsmath,amssymb}
\usepackage{algorithm}
\usepackage{algorithmic}
\usepackage{enumitem}
\usepackage{amsfonts}
\usepackage{graphicx}
\usepackage{epstopdf}
\usepackage{url}

\newtheorem{theorem}{Theorem}
\newtheorem{lemma}{Lemma}
\newtheorem{proposition}{Proposition}
\newtheorem{corollary}{Corollary}

\usepackage{amsopn}

\DeclareMathOperator{\ex}{\mathbb{E}}
\DeclareMathOperator{\pr}{\mathrm{Pr}}

\newcommand\diff{\mathrm{d}}

\newcommand\influ{\mu}
\newcommand\qhat{\hat{q}}
\newcommand\rhohat{\hat{\rho}}
\newcommand\influhat{\hat{\influ}}

\newcommand\Qhat{\hat{Q}}
\newcommand\qbar{\bar{q}}

\newcommand\Scal{\mathcal{S}}

\newcommand\dbar{\bar{d}}

\newcommand\alphabar{\bar{\alpha}}
\newcommand\alphaul{\underline{\alpha}}

\title{Influence Prediction for Continuous-Time Information Propagation on Networks}

\author{
Shui-Nee Chow\thanks{School of Mathematics, Georgia Institute of Technology, Atlanta, GA
	(\texttt{chow@math.gatech.edu}).}
	\and
	Xiaojing Ye\thanks{Department of Mathematics \& Statistics, Georgia State University, Atlanta, GA (\texttt{xye@gsu.edu}).}
	\and
	Hongyuan Zha\thanks{School of Computational Science \& Engineering, College of Computing, Georgia Institute of Technology, Atlanta, GA (\texttt{zha@cc.gatech.edu}).}
	\and
	Haomin Zhou\thanks{School of Mathematics, Georgia Institute of Technology, Atlanta, GA
	(\texttt{hmzhou@math.gatech.edu}).}
}
\date{}

\begin{document}
\maketitle

\begin{abstract}
We consider the problem of predicting the time evolution of influence,
defined by the expected number of activated (infected) nodes, given
a set of initially activated nodes on a propagation network. 
To address the significant computational challenges of this problem on large heterogeneous networks, 
we establish a system of differential equations governing the 
dynamics of probability mass functions on the state graph where each node
lumps a number of activation states of the network, which can be considered
as an analogue to the Fokker-Planck
equation in continuous space. We provides several methods to estimate the
system parameters which depend on the identities of the initially active nodes,
network topology, and activation rates etc.
The influence is then estimated by the solution
of such a system of differential equations.
Dependency of prediction error on parameter estimation is established.  
This approach gives rise to a class of novel and scalable algorithms that work 
effectively for large-scale and dense networks.
Numerical results are provided to show the very promising performance
in terms of prediction accuracy and computational efficiency of
this approach.
\smallskip


\end{abstract}



\section{Introduction}
\label{sec:intro}

Viral signal propagation on large heterogeneous networks is an emerging research subject
of both theoretical and practical importance.
Influence prediction is one of the most fundamental problems
about propagation on networks, and it has been arising from many
real-world applications of significant societal impact, such as news spread
on social media, viral marketing, computer malware detection, 
and epidemics on heterogeneous networks. For instance,
when considering a social network formed by people such as that of 
Facebook or Twitter, the viral signal can be a tweet or a trendy topic being 
retweeted by users (nodes) on the network formed by their followee-follower relationships.
We call a user activated if he/she participates to tweet, and the followers
of this user get activated if they retweet his/her tweet later, thus
the activation process gradually progresses (propagates) and the tweet spreads out.
A viral signal can also be a new electronic gadget that 
finds wide-spread adoption in the user population through a 
word-of-mouth viral marketing process \cite{Leskovec:2009a,Leskovec:2007b,Ryu:2012a},
and a user is called activated when he/she adopts this new gadget. 
Influence prediction is to quantitatively estimate how influence, 
defined by the expected number of activated nodes, evolves over time
during the propagation when a specific set (called source set) of nodes are initially activated.

Influence prediction is also the most critical step in solving problems 
arising from many important downstream applications such as 
influence maximization \cite{Cohen:2014a,Gomez-Rodriguez:2012c,Kempe:2003a,Wang:2012a}
and outbreak detection \cite{Cui:2013a,Leskovec:2007b}.
For instance, in influence maximization, the goal is to select the source node set
of a given size from the propagation network such that its influence is maximized at a prescribed time. 
Obviously, influence prediction serves as the most fundamental subroutine 
in the computation, and the quality of influence maximization heavily depends on the
accuracy of influence prediction.

\subsection{Problem description}\label{subsec:problem_description}
The influence prediction problem can be formulated as follows.
Let $G=(V,E)$ be a given network (directed graph) with node (vertex) set 
$V=\{i:1\leq i \leq K\}$ and edge set $E\subset V\times V$. 
We denote $N_i^{\mathrm{in}}:=\{j:(j,i)\in E\}$
and $N_i^{\mathrm{out}}:=\{j:(i,j)\in E\}$.
A piece of information on $G$ can spread or propagate 
from an active node $i$ to every inactive $j\in N_i^{\mathrm{out}}$, and once succeeded, 
node $j$ becomes active and starts to propagate information to
inactive nodes in $N_j^{\mathrm{out}}$, and so on.
Once $i$ is activated, the time elapsed for $i$ to activate $j$
follows certain probability. 
In addition to node-to-node activations, the nodes may also have the ability of
self-activation. Namely, an inactive node $i$ 
can be self-activated regardless of having any activated nodes in $N_i^{\mathrm{in}}$. 
We assume the standard case that activated nodes cannot be activated again
unless recovery scenario is considered.
At any time, each node is in one of two states: inactive (susceptible) or
active (infected). This model is called susceptible-infected (SI) 
model in classical mathematical epidemics theory.
However, unlike most of existing works in this field, we here focus on efficient 
computational method for influence prediction in the following settings
due to practical concerns in real-world social networking applications.
\begin{itemize}[leftmargin=7mm]
\item The network $G$ is deterministically heterogeneous.
This is significantly different from the case of classical SI model 
in mathematical biology/epidemics theory
which does not consider contact network at individual level.
Our network is also different from heterogeneous but statistically homogeneous
networks considered in statistical physics literature,
where nodes can be partitioned into multiple categories according to certain properties, e.g.,
degrees, and the nodes within each category can be treated equivalently. 
In our case, the edges are explicitly given in the static network
and the activation times have different but fixed distributions.

\item Quantitative estimate of influence for time $t$ before equilibrium.
Note the equilibrium state of SI model on network is trivial:
all nodes that can be reached from the source set will be infected
as time tending to infinity. However, practical interests often
lie in influence before equilibrium. For example, merchants would like
to know how many people will be influenced by commercial advertisement
within one month rather than three years later. In this case, we need to 
estimate the time evolution of influence in early to middle stage
where the propagation is still in nonequilibrium state.
\end{itemize}
Our discussion also includes the case of self-activation where
the unactivated nodes can activate themselves automatically.
If infected nodes can recover, become susceptible and prune to future infection,
then the model is called susceptible-infected-susceptible (SIS), for which we only
provide a brief discussion near the end of this paper.

Given network $G=(V,E)$, the stochastic propagation process is determined
by the distribution of the activation times between nodes. 
In this paper, we mainly consider the model with activation times exponentially
distributed which is widely
used in classical epidemic study and social network.
In this model, the time for a just activated node
$i$ to activate each unactivated $j$ in $N_i^{\mathrm{out}}$, denoted
by $t_{i,j}$, follows $\exp(\alpha_{ij})$ distribution
(here
$t$ follows $\exp(\alpha)$ distribution if the probability density function of $t$ is
$p_t(\tau)=\alpha e^{-\alpha \tau}$ for $\tau\geq0$),
and is independent 
of any other $t_{i',j'}$ where $i\neq i'$ and/or $j\neq j'$. 
Here $\alpha_{ij}>0$ indicates the instantaneous activation rate of $j$ by
$i$. If $(i,j) \notin E$, we set $\alpha_{ij}=0$ by convention.
Note that exponential random variable following $\exp(\alpha)$ has mean
$1/\alpha$, therefore, the larger $\alpha_{ij}$ is, the faster $i$ can
activate $j$ on expectation. Hence $\alpha_{ij}$ can be interpreted
as the impact level (weight) of $i$ on $j$.
Similarly, the time for an inactive node $i$ to get self activated
follows $\exp(\beta_i)$ for some $\beta_i>0$. When recovery scenario is considered, an activated
node $i$ can recover in time following $\exp(\gamma_i)$ for some $\gamma_i>0$
since it is activated.

The continuous-time propagation model with heterogeneous activation rates
appears suitable for a great number of real-world applications and 
has been advocated by many recent works
\cite{Du:2013a,Gomez-Rodriguez:2012c,Pastor-Satorras:2015a,
Van-Mieghem:2013a,Van-Mieghem:2009a}.
In addition, this model yields a time-homogeneous Markov 
propagation process so that numerical simulations can be implemented
in a straightforward manner and some theoretical analysis of the
algorithm can be carried out.
Therefore, we focus on the development of the algorithm on this
propagation model, and evaluate the performance numerically to obtain
references worthy of trust through a large amount of Monte Carlo simulations.
However, the general framework using Fokker-Planck equations - the main strategy in the current work -
as well as the error estimations developed in Section \ref{subsec:error}
apply to any propagation models (e.g., activation time not exponentially
distributed, such as Hawkes processes) on networks.

To estimate time evolution of influence of a
source set $S$, we define a single stochastic process $N(t;S)$ as the number of 
activated nodes at time $t$ when the source set is $S$.
Then we directly compute the probability distribution of $N(t;S)$:
\begin{equation}\label{eq:rho}
\rho_k(t;S):=\pr (N(t;S)=k),\quad \text{for}\ k=0,1,\dots,K.
\end{equation}
The influence, defined by the expected number of activated nodes 
at time $t$, can therefore be calculated easily by
\begin{equation}\label{eq:influ}
\influ(t;S) = \ex [N(t;S)]=\sum_{k=0}^K k\rho_k(t;S),
\end{equation}
where $K:=|V|$ is the size of the network.

The main focus of this paper is to establish a general framework for 
computing (predicting) influence $\influ(t;S)$
based on \eqref{eq:rho} and \eqref{eq:influ} for any given source set $S$.
More precisely, we build the system of equations for the time evolution of
$\{\rho_k(t;S): 0\leq k\leq K\}$, analyze its properties, 
estimate the parameters in the equations,
and solve for all $\rho_k(t;S)$ to predict the influence $\influ(t;S)$ in \eqref{eq:influ}
for all $t$. Since the source $S$ is arbitrarily set in advance,
we drop the symbol $S$ in the derivation hereafter for notation simplicity.

The idea of deriving evolution equations of
$\{\rho_k(t): 0\leq k\leq K\}$ is closely related to the theory of Fokker-Planck
equation. In continuous space $\mathbb{R}^n$, consider a classical stochastic process
$X(t)$ that stands for the location of a particle at time $t$.
Let $\rho(x,t)$ denote the
probability density that $X(t)$ is located at $x\in \mathbb{R}^n$ at time $t$,
then $\rho(x,t)$ evolves over time with a constraint $\int_{\mathbb{R}^n}\rho(x,t)\diff x=1$ at every $t$.
The Fokker-Planck equation, also known as the forward Kolmogorov equation, 
is a deterministic partial differential equation governing the 
time evolution of $\rho(x,t)$.
For example, if $X(t)$ moves according to a 
stochastic differential equation (SDE)
$\diff X(t) = - \nabla \Psi(X(t)) \diff t + \sqrt{2\beta} \diff W(t)$
where $W(t)$ is an $n$-dimensional Brownian motion and
$\Psi(\cdot)$ is a scalar-valued potential function,
then the Fokker-Planck equation of $\rho(x,t)$ is
$\partial_t \rho(x,t) = \nabla \cdot (\nabla \Psi(x){ \rho}(x,t)) + \beta \Delta { \rho}(x,t)$. 
Here $\Delta { \rho}(x,t)$ corresponds to the $W(t)$ term 
in the SDE and is called the diffusion term,
and $\nabla \cdot (\nabla \Psi(x){ \rho}(x,t))$ is the drift term.
Note that the statistics of $X(t)$ can be completely determined by
the solution $\rho(x,t)$ of the Fokker-Planck equation.

Likewise, the probabilities $\{\rho_k(t):0\leq k\leq K\}$ in our approach also evolve over time
with $\sum_{k=0}^K \rho_k(t)=1$ at all $t$.
The time evolution of $\rho_k(t)$ is also governed by certain
Fokker-Planck equation which is now a system of deterministic ordinary differential equations
since the state space is discrete $N(t)=0,1,\dots,K$ rather than continuous $\mathbb{R}^n$.
In recent years, there have been growing research interests in
general graph-based Fokker-Planck equations 
to study problems related to optimal transport on finite graphs \cite{Che:2016a,Erbar:2012a,Chow:2011b}.
In the present paper, however, our goal is to find the Fokker-Planck equation
that governs $\rho_k(t)$ in \eqref{eq:rho}, and solve for these $\rho_k(t)$ to
obtain influence $\influ(t)$ using \eqref{eq:influ}.
We also analyze how the coefficient errors in Fokker-Planck equation affect 
accuracy of predicting $\influ(t)$ using this approach.

\subsection{Related Work}
\label{sec:related}

Previous study of influence estimation on networks is mainly restricted to statistically 
homogeneous and well-mixed populations, particularly in the context of statistical properties of dynamical
processes on complex networks in physics. 
A comprehensive survey is provided in \cite{Pastor-Satorras:2015a}.
The typical approach is based on mean-field approximation (MFA) 
to establish a system of differential equations for the compartment
model which groups nodes with statistically identical properties into one. 
For example, degree-based MFA groups nodes of the same degree which are considered to have 
identical behavior statistically, 
and hence can significantly reduce the size of the system \cite{Boguna:2002a}.
Pair approximation includes the joint distribution in the system of
equations, which essentially applies moment closure after the joint
distribution of paired nodes, and is shown to have improved accuracy over 
standard MFA \cite{Boccaletti:2014a,Demirel:2014a,Newman:2010a}.
Other generalization and improvements of MFA and pair approximation
using compartment models and motif expansions
can be found in \cite{Demirel:2014a,Lindquist:2011a,Marceau:2010a,Miller:2014a,Taylor:2014a},
and references therein.

As noted in Section \ref{subsec:problem_description}, our focus in this paper
is instead on influence prediction (estimation) on \textit{deterministically heterogeneous networks} 
particularly in \textit{non-equilibrium} stage, which is significantly different from
existing works including those mentioned above.
For influence prediction in this setting, prototype MFA for Markov propagation model developed 
in \cite{Kephart:1991a,Wang:2003a} is generalized to arbitrary network topology \cite{Van-Mieghem:2009a},
and then further extended to inhomogeneous activation and
recovery rate between nodes \cite{Van-Mieghem:2013a}. The model adopts the MFA and first-order
moment closure, i.e., substituting the
joint distribution of two activated neighbor nodes by product of marginal distributions 
for individual nodes, to retain a feasible size of the derived system of 
differential equations. A second- (or higher-) order moment closure can be considered 
but the limitations and instability are discussed in \cite{Cator:2012a}. 
Assuming absence of recovery, the 
exact solution is available due to the Markov property of propagation,
however, its computational complexity increases drastically in terms of
size and density of general networks \cite{Gomez-Rodriguez:2012c}.
As an alternative to solving for influence based on evolution equations,
methods based on sampling propagations (also called cascades) and statistical learning technique
are also developed, but often posing various requirements
on input data and output results. For instance, a scalable computational method based on 
learning the coverage function of each node based on sampling and kernel estimation is developed,
which can can only predict the influence at a prescribed time \cite{Du:2013a}. The work is further extended
to estimate the time-varying intensity of propagation using
similar coverage function idea \cite{Du:2014a}. 
Learning-based methods are usually companioned with a great amount
of accuracy analysis based on classical theory of sampling complexity.
However, the major problem with learning-based approaches is
in the use of large amount of samplings to ensemble the unknown
function or probability of interests but lack of
a comprehensive understanding of the underlying dynamics
and unique properties associated with the stochastic
propagation on networks. Moreover, learning-based methods can
have special assumptions on data which may not be realistic in real-world
applications. 
To achieve moderate accuracy level in large-scale and complex network,
learning-based methods require extensive amount of sampling/simulations
which causes significant computational burden and hinders their applicability
in real-world problems.

\section{Proposed Method}
In this section, we first derive the Fokker-Planck equation for the probabilities
$\rho_k(t)$ of $N(t)$ for the propagation model with exponentially distributed
activation times. We provide two effective methods to estimate the coefficients in the 
Fokker-Planck equation for large heterogeneous networks. 
Then we establish the relation between the estimation error of the
coefficients in the Fokker-Planck equation and the accuracy in 
the predicted influence
for general propagation models using our approach. 

\subsection{The Fokker-Planck equation of $\rho_k(t)$}
Let $G=(V,E)$ and $\{\alpha_{ij}:(i,j)\in E\}$
(and $\{\beta_i:i\in V\}$ for self-activation and $\{\gamma_i:i\in V\}$ for recovery) 
be given and the source set $S$ be chosen arbitrarily.
The number of activated nodes, $N(t)$, has $K+1$ states corresponding
to $N(t)=0,1,\dots,K$. 
Let $M_k$ denote the state that $N(t)=k$ nodes in $G$ are activated.
Then, for the general SIS model with self-activation and recovery, 
the transitions between states of $N(t)$ can be illustrated as follows,
\begin{equation}\label{eq:FPEstruct}
\boxed{M_0} 
\rightleftharpoons \dots \rightleftharpoons
\boxed{M_{k-1}} \underset{r_k(t)}{\stackrel{q_{k-1}(t)}{\rightleftharpoons}} \boxed{M_k}
\underset{r_{k+1}(t)}{\stackrel{q_k(t)}{\rightleftharpoons}} \boxed{M_{k+1}}
\rightleftharpoons \dots \rightleftharpoons \boxed{M_{K}}
\end{equation}
Here, $q_k(t)$ is the transition rate from $M_k$ to $M_{k+1}$
and $r_k(t)$ is the rate from $M_k$ to $M_{k-1}$ at time $t$,
and they depend on the structure of $G=(V,E)$, the activation parameters $\alpha_{ij}$
(and $\beta_i$ and $\gamma_i$ for self-activation and recovery respectively),
and the source set $S$. 

Recall that $\rho_k(t)$ is the probability of $N(t)$ being in state $M_k$
according to definition \eqref{eq:rho}.
Therefore, the time evolution of $\rho_k(t)$ is governed by
the discrete Fokker-Planck equation with these $q_k(t)$ and $r_k(t)$ to be determined:
\begin{align}
\rho_0'(t) & = -q_0(t) \rho_0(t) + r_1(t) \rho_1(t), \nonumber \\
\rho_k'(t) & = q_{k-1}(t)\rho_{k-1}(t)-[q_{k}(t)+
r_{k}(t)]\rho_{k}(t) + r_{k+1}(t)\rho_{k+1}(t),\ 0< k <K, \label{eq:FPEform}\\
\rho_K'(t) & = q_{K-1}(t)\rho_{K-1}(t)- r_{K}(t) \rho_{K}(t). \nonumber
\end{align}
To rewrite \eqref{eq:FPEform} into a concise matrix formulation, we 
define two $(K+1)\times(K+1)$ matrices $Q(t)$ and $R(t)$ as follows: 
\begin{align}
[Q(t)]_{j,j} &= -q_{j-1}(t), \quad [Q(t)]_{j,j+1}=q_{j-1}(t),\quad j=1,\dots,K \label{eq:Q} \\
[R(t)]_{j,j} &= -r_{j-1}(t), \quad [R(t)]_{j,j-1}=r_{j-1}(t),\quad j=2,\dots,K+1. \label{eq:R}
\end{align}
and all other entries are zeros. Here $[P]_{j,l}$ stands for the $(j,l)$-th entry of matrix $P$.
Note that only the diagonal and superdiagonal (subdiagonal) entries of $Q(t)$ ($R(t)$)
are nonzeros, and $[Q(t)]_{K+1,K+1}=[R(t)]_{1,1}=0$ for all $t$.
With matrices $Q(t)$ and $R(t)$ given above, we define
a row $(K+1)$-vector $\rho(t):=(\rho_0(t),\rho_1(t),\dots,\rho_K(t))$ and
rewrite \eqref{eq:FPEform} as 
\begin{equation}\label{eq:FPEform_mtx}
\rho'(t)=\rho(t)[Q(t)+R(t)].
\end{equation}

The system \eqref{eq:FPEform_mtx} is consistent with the nature of process $N(t)$ in \eqref{eq:FPEstruct}
with a tridiagonal transition matrix $Q(t)+R(t)$.
The initial value $\rho(0)$ can be easily determined given $S$: let $|S|$ denote the 
cardinality of $S$, then $\rho(0)$ is a binary $(K+1)$-vector such that
$\rho_{|S|}(0)=1$ and $\rho_k(0)=0$ for all $k\neq |S|$.
Therefore, we can solve \eqref{eq:FPEform} for $\rho(t)$ 
to obtain influence $\influ(t)$ based on \eqref{eq:influ}
once the transition rates $q_k(t)$ and $r_k(t)$ are determined.
The following subsection is devoted to the estimation of these rates.

\subsection{Estimation of transition rates $q_k(t)$ and $r_k(t)$}\label{subsec:rates}

Recall that $q_k(t)$ stands for the transition rate of $N(t)$ from $M_k$
to $M_{k+1}$ as shown in \eqref{eq:FPEstruct}. 
Namely, $q_k(t)$ is the instantaneous rate for the $(k+1)$-th node
to be activated given that there are currently $k$ activated node
(with numerous possible choices of such $k$ nodes in $V$ and
$q_k(t)$ aggregates all the information) at time $t$.
Similarly, $r_k(t)$ is the instantaneous rate for any of these
$k$ activated nodes to get recovered. Therefore, we focus
on the estimation of $q_k(t)$ and a similar derivation can be easily 
carried out for $r_k(t)$. 

The estimation of rate $q_k(t)$ consists of two factors: 
(i) the identities of the $k$ currently activated nodes,
and (ii) the instantaneous activation rate imposed by these $k$ nodes 
to all the unactivated nodes at the time $t$.
For factor (ii), the propagation model with exponentially distributed activation times 
yield constant instantaneous rates if the identities of the $k$ nodes are given.
For factor (i), $q_k(t)$ need to aggregate all ${K \choose k}$ 
possible combinations of $k$ activated nodes. 
The following theorem provides the compositions of $q_k(t)$ and $r_k(t)$.
Here we call $U$ activated if all nodes in $U$ are activated and the others in $U^c=U\setminus V$ are 
unactivated.
The proof frequently calls two simple facts about selection probability given multiple
instantaneous rates, which we provide as Propositions \ref{prop:minT} and \ref{prop:mixrate} 
in the Appendix for completeness.
\begin{theorem}\label{thm:qrest}
For every $k=0,1,\dots,K$, let $\Scal_k:=\{U\subset V: |U|=k\}$ be the collection of all subsets of size $k$ in $V$.
Let $\pr(t; U)$ be the probability that $U$ is activated among those in $\Scal_k$,
and define
\begin{equation}\label{eq:alpha}
\alpha(U)=\sum_{i\in U}\sum_{j\in N_i^{\mathrm{out}}\cap U^c} \alpha_{ij},\quad
\beta(U) = \sum_{i \in U} \beta_i,\quad
\gamma(U)=\sum_{i \in U} \gamma_i.
\end{equation} 
Then the transition rates $q_k(t)$ and $r_k(t)$ in \eqref{eq:FPEform} are given by
\begin{equation}\label{eq:rates}
q_k(t)=\sum_{U\in \Scal_k} \left[\alpha(U)+\beta(U^c)\right]\pr(t; U)\ \ \mbox{and }\ 
r_k(t)=\sum_{U\in \Scal_k}\gamma(U)\pr(t; U).
\end{equation}
\end{theorem}
\begin{proof}
Suppose nodes in $U\subset \Scal_k$ are currently activated,
then these nodes tend to activate their neighbors still in $U^c$ independently
and simultaneously. 
More precisely, node $i\in U$ imposes a node-to-node activation rate
to each of its unactivated neighbor $j\in N_i^{\mathrm{out}}\cap U^c$ independently and simultaneously,
and hence total rate is $\sum_{j\in N_{i}^{\mathrm{out}}\cap U^c} \alpha_{ij}$.
Therefore, the combined instantaneous rate of all nodes in $U$ for node-to-node activation 
is given by $\alpha(U)$ in \eqref{eq:alpha} according to Proposition \ref{prop:minT}.
Meanwhile, each node $i$ in $U^c$ tends to be self activated with rate $\beta_i$
and hence their total instantaneous self-activation rate is $\beta(U)$ given in \eqref{eq:alpha}.
As $\pr(t;U)$ is the probability that nodes in $U$ are activated, we
obtain the total instantaneous rate $q_k(t)$ as for $N(t)$ to transit from
state $M_k$ to $M_{k+1}$ as \eqref{eq:rates} according to Proposition \ref{prop:mixrate}.
The derivation for $r_k(t)$ in \eqref{eq:rates} follows similarly.
\end{proof}

Note that there are $|\Scal_k|={K\choose k}$ possible combinations $U$ and
$\sum_{U\in \Scal_k} \pr(t;U)=1$ for all $t$. 
Moreover, $q_k(t)$ is a convex combination of the instantaneous
activation rates $\alpha(U)+\beta(U^c)$
with weights given by $\pr(t;U)$ according to Theorem \ref{thm:qrest}. 
Hence $q_k(t)$ is closer to the $\alpha(U)+\beta(U^c)$ with larger $\pr(t;U)$.
The composition of $r_k(t)$ in Theorem \ref{thm:qrest} has similar interpretation.
Although it is not practical to obtain the probability $\pr(t;U)$ 
for all $U$, Theorem \ref{thm:qrest} suggests that we can approximate 
$q_k(t)$ and $r_k(t)$ using the activation and recovery rates
of those $U$ with large weight $\pr(t;U)$. 

We now present two estimation methods and practical implementations 
using this idea for the case without
self-activation and recovery (which essentially yields the standard
susceptible-infection (SI) propagation model) on heterogeneous networks. 
In this case, we have
$q_k(t)=\sum_{U\in \Scal_k} \alpha(U)\pr(t;U)$ and $r_k(t)=0$.
Therefore, the key is to approximate $q_k(t)$ using $\alpha(U)$ 
of few $U$ with the largest probabilities $\pr(t;U)$.

\textit{Estimate $q_k$ based on the shortest distance.} 
For every $k=1,\dots,K$, we can easily determine a combination $U_k^*$ 
with large $\pr(t;U)$ over all $U$ in $\Scal_k$ as follows:
recall that the expected time for node $i$ to activate $j\in N_i^{\mathrm{out}}$ is $1/\alpha_{ij}$,
it is therefore natural to define the distance from $i$ to $j$ as $D(i,j):=1/\alpha_{ij}$, 
which can also be generated to the distance from set $S$ to a node $j$
as $D(S,j):=\min_{i\in S}D(i,j)$.
Due to independency of all node-to-node activations and property of
exponential distributions, 
the set $U_k^*$ consisting of the $k$ nodes with shortest distance to source $S$
has larger 
probability to be activated first among those in $\Scal_k$. 
In practical implementation, we apply Dijkstra's method \cite{Sniedovich:2006a}
on the weighted graph $G$ with edge weights given by $1/\alpha_{ij}$ and
origin $S$, and 
then sort the nodes as $i_1,i_2,\dots,i_K$ with ascending distance from source $S$, i.e.,
$D(S,i_1)\leq D(S,i_2)\leq \dots \leq D(S,i_K)$ (if $i\in S$ then $D(S,i)=0$)
and set $U_k^*=\{i_1,\dots,i_k\}$ for $k=1,\dots,K$. Then we approximate
$q_k(t)$ by $\qhat_k(t)=\alpha(U_k^*)$, which remains
as constant for all $t$ once the source set $S$ is given.
This method is referred to as \texttt{FPE-dist} in the numerical experiments.


\textit{Estimate $q_k$ based on the largest overall probabilities.}
To refine the approximation using single $U_k^*$ in \texttt{FPE-dist},
we can estimate $q_k(t)$ using multiple combinations $U$
with the largest probabilities. 
For a fixed $S$, 
we employ the the following recursive method to determine the sets $\{U_k^{1},\dots,U_k^{m_k}\}\subset \Scal_k$
to be used in calculation of $q_k$ in \eqref{eq:rates}
for $k=1,2,\dots,K-1$. Here $m_k$ is a user-customized number of $k$-combinations 
selected from $\Scal_k$ (larger $m_k$ yields
more accurate approximation to $q_k$ at the expense of higher computation complexity.)
Suppose we have already obtained $U_k^{1},\dots,U_k^{m_k}$ for $k$
such that $\pr(U_k^{1})\geq \cdots \geq \pr(U_k^{m_k})$, 
our next step is to obtain $U_{k+1}^{1},\dots,U_{k+1}^{m_{k+1}}$.
To this end, we proceed with previous $U_k^l$ in the order of $l=1,\dots,m_k$
and compute $\alpha(j|U_k^l)$ for every neighbor node $j$ of $U_k^l$
(by neighbor $j$ of a subset $U$ we meant that $j\in N_i^{\mathrm{out}}$ for some $i\in U$,
and $\alpha(j|U):=\sum_{i\in U\cap N_j^{\mathrm{in}}}\alpha_{ij}$ is the total activation rate
imposed to $j$ by nodes in $U$.)
By Proposition \ref{prop:minT}, neighbor $j$ of $U_k^l$ will be activated before other neighbors $j'$
with probability $\pr(j|U_k^l)=\alpha(j|U_k^l)/\sum_{j'} \alpha(j'|U_k^l)$ where the summation in the denominator
is over all neighbors $j'$ of $U_k^l$.
Therefore, $\pr(U)=\pr(j|U_k^l)\pr(U_k^l)$ for $U:=U_k^l\cup\{j\}$
and $T(U_k^l)\leq t_j$. Note that each neighbor $j$ of $U_k^l$ yields such a $U$
of size $k+1$.
All these $U$'s are then candidates for $U_{k+1}^{1},\dots,U_{k+1}^{m_{k+1}}$
later. We proceed with each $U_k^l$ in the aforementioned way for $l=1,\dots,m_k$
and obtain a number of sets $U$'s with probabilities $\pr(U)$.
Note that if two or more of these $U$'s are identical, then we keep only one of them
and merge their probabilities $\pr(U)$. 
Then we sort these $U$'s with $\pr(U)$ in descending order and only keep the first $m_{k+1}$
as $U_{k+1}^{1},\dots,U_{k+1}^{m_{k+1}}$.
By this method, we are likely (but not guaranteed) to maintain a list $\{U_k^1,\dots,U_k^{m_{k}}\}$
with largest probabilities among all those in $\Scal_k$ for each $k$. 
Then we approximate $q_k(t)$ by $\qhat_k(t):=\sum_{l=1}^{m_k}\alpha(U_k^{l})\pr(U_k^{l})$
which is again constant for all $t$.
This method essentially constructs a branching tree with $K+1$ layers, where layer $k$
consists of $m_k$ nodes $U_k^1,\dots,U_k^{m_k}$ each having a relative probability in its layer,
and others with small probabilities in $\Scal_k$ are removed 
so that computation complexity is maintained within a feasible scale. 
We refer this method to as \texttt{FPE-tree} in the numerical experiments.

Once we obtained the estimate $\qhat_k(t)$, the last step is
to solve the Fokker-Planck equation $\rho'(t)=\rho(t)Q(t)$ numerically. 
There are two straightforward methods to compute $\rho(t)$:
the Runge-Kutta method which can handle time varying $Q(t)$
and very large $K$ (with computation complexity $O(K)$)
but needs to proceed the computation starting from $t=0$;
and direct computation of $\rho(t)=\rho(0)e^{\int_0^tQ(s)ds}$
with bidiagonal matrix $Q$. In particular, if $Q$ is constant,
then the computation $\rho(t)=\rho(0)e^{tQ}$ is very fast
using matrix exponential \cite{Moler:2003a,Sidje:1998a,Xue:2013a}
and can be directly done for any specific $t>0$ rather than from $t=0$.

The steps for influence prediction using Fokker-Planck equation 
\eqref{eq:FPEform} is summarized in Algorithm \ref{alg:FPE}.
For completeness we include the self-activation and recovery rates.

\begin{algorithm}
\caption{Influence prediction based on Fokker-Planck equation \eqref{eq:FPEform_mtx}}
\label{alg:FPE}
\begin{algorithmic}[1]
\STATE{\textbf{input} $G=(V,E)$, $\{\alpha_{ij},\beta_i,\gamma_i:(i,j)\in E, i\in V\}$. Give source set $S\subset V$.}
\STATE{Estimate $\{q_k(t),r_k(t):t\geq0\}$ defined in \eqref{eq:rates} and form matrices $Q(t),R(t)$ as in \eqref{eq:Q}-\eqref{eq:R}.}
\STATE{Solve $\rho'(t)=\rho(t)[Q(t)+R(t)]$ with initial $\rho(0)$ to obtain $\rho(t)$.}
\RETURN Output influence $\influ(t)=\sum_{k=0}^{K}k \rho_k(t)=\rho(t)(0,1,\dots,K)^T$. 
\end{algorithmic} 
\end{algorithm}

\subsection{Error estimate for influence prediction}\label{subsec:error}
In this section, we conduct error analysis of the proposed influence prediction method.
For simplicity, we consider the case without recovery scenario, and assume that
the propagation starts with self-activation, i.e., $\rho(0)=(1,0,\dots,0)\in\mathbb{R}^{K+1}$,
since derivations generalize to other initials trivially.
It is worth noting that the results obtained in this section apply to 
any propagation model.

We first observe that the solution $\rho(t)=(\rho_0(t),\dots,\rho_K(t))$ of $\rho'(t)=\rho(t)Q(t)$ 
with initial value $\rho(0)=(1,0,\dots,0)$ is
\begin{align}\label{eq:FPEsol}
\rho_0(t) &= e^{-\int_0^t q_0(s)\diff s}, \nonumber\\
\rho _{k+1}\left( t\right) &=\int ^{t}_{0}\rho _{k}\left( s\right) q_{k}\left( s\right) e^{-\int ^{t}_{s}q_{k+1}\left( u\right) \diff u}\diff s,\ \mbox{for }k=0,1,\dots,K-2, \\
\rho_K(t) &= \int_0^t \rho_{K-1}(s)q_{K-1}(s)\diff s. \nonumber
\end{align}
Now, for every $k=0,1,\dots,K-1$, let $Q_k$ denote the perturbed rate matrix as
\begin{equation}
Q_k(t)=
\begin{pmatrix}
& \ddots & \ddots & & & \\
& & -q_{k-1}(t) & q_{k-1}(t) & 0 & \\
& & 0 & -\qhat_{k}(t) & \qhat_{k}(t) & & \\
& & & & \ddots & \ddots
\end{pmatrix}
\end{equation}
That is, $Q_k(t)$ differs from the orignal $Q(t)$ by replacing $q_j(t)$
with $\qhat_j(t)$ for $j=k,k+1,\dots,K-1$.
Then we have the following lemma that relates the error in solution $\rho(t)$ to
the error in estimating $q_k(t)$.

\begin{lemma}\label{lemma:rho_error}
Let $\epsilon\in(0,1)$, and $\rho$ and $\rhohat$ be the solutions of $\rho'(t)=\rho(t)Q_{k+1}(t)$ and
$\rhohat'(t)=\rhohat(t)Q_{k}(t)$, respectively. Denote $\delta_k(t):=|\qhat_k(t)-q_k(t)|/q_k(t)$. 
If $\alphabar>0$ is the upper bound of all activation rates between nodes in $G=(V,E)$ 
and that
\begin{equation}\label{eq:delta_bound}
\delta_k(t)\leq \min\left\{\frac{\log(1+\frac{\epsilon}{2})}{\alphabar k t\min(\dbar,K-k)}, \frac{\epsilon}{2+\epsilon}\right\}
\end{equation} 
where $\dbar=\max\{|N_i^{\mathrm{out}} |:i\in V\}$, then 
$\rho_j(t)=\rhohat_j(t)$ for $j=0,1,\dots,k-1$ and
$|\rhohat_j(t)-\rho_j(t)|/\rho_j(t)\leq \epsilon$ for $j=k,\dots,K$ and all $t>0$.
Moreover, $|\influhat(t)-\influ(t)|/\influ(t)\leq \epsilon$ for all $t$.
\end{lemma}

\begin{proof}
If $k>0$, from the solution formulation \eqref{eq:FPEsol}, we know that
$\rho_j(t)=\rhohat_j(t)$ for all $t$ and $j=0,1,\dots,k-1$. Furthermore, there are
\begin{align}
\rho _{k}(t) =\int ^{t}_{0}\rho_{k-1}(s) q_{k-1}(s) e^{-\int ^{t}_{s}q_{k}(u) \diff u}\diff s, \label{eq:rhok}\\
\rhohat _{k}(t) =\int ^{t}_{0}\rho_{k-1}(s) q_{k-1}(s) e^{-\int ^{t}_{s}\qhat_{k}(u) \diff u}\diff s. \label{eq:rhohatk}
\end{align}
Since $q_k(t)\leq \alphabar k \min(\dbar,K-k)$ and \eqref{eq:delta_bound}, there are
$\int_0^t\delta_k(s)q_k(s)\diff s\leq \log(1+\frac{\epsilon}{2})$ and
\begin{equation}
\left| e^{-\int_s^t(\qhat_k(u)-q_k(u))\diff u}-1\right| \leq e^{\int_s^t\delta_k(u)q_k(u)\diff s}-1
\leq e^{\int_0^t\delta_k(u)q_k(u)\diff s}-1\leq \frac{\epsilon}{2}
\end{equation}
for all $s\in(0,t)$. Therefore, from \eqref{eq:rhok} and \eqref{eq:rhohatk} there is
\begin{align*}
\frac{|\rhohat_k(t)-\rho_k(t)|}{\rho_k(t)}
& \leq \frac{1}{\rho_k(t)}\int_0^t \rho_{k-1}(s) q_{k-1}(s) e^{-\int ^{t}_{s}q_{k}(u)\diff u} \left|e^{-\int_s^t(\qhat_k(u)-q_k(u))\diff u}-1\right|\diff s \\
& \leq \frac{\epsilon}{2\rho_k(t)}\int_0^t \rho_{k-1}(s) q_{k-1}(s) e^{-\int ^{t}_{s}q_{k}(u)\diff u}\diff s=\frac{\epsilon}{2}.
\end{align*}
If $k=0$, then $\rho_0(t)=e^{-\int ^{t}_{0}q_{0}(s) \diff s}$ and $\rhohat_0(t)=e^{-\int ^{t}_{0}\qhat_{0}(s) \diff s}$
and one can check that this inequality still holds.

As $|\rhohat_k(t)-\rho_k(t)|/\rho_k(t)\leq\epsilon/2$, there is $\rhohat_k(t) \leq (1+\epsilon/2)\rho_k(t)$ and hence
\begin{multline}
|\rho_k(t)q_k(t)-\rhohat_k(t)\qhat_k(t)| \leq |\rho_k(t)-\rhohat_k(t)|q_k(t) + \rhohat_k(t)|q_k(t)-\qhat_k(t)) | \\
\leq  \frac{\epsilon}{2}\rho_k(t)q_k(t) + \left(1+\frac{\epsilon}{2}\right)\rho_k(t)\delta_k(t)q_k(t)\leq \epsilon \rho_k(t)q_k(t)
\end{multline}
for all $t\geq0$, where we used the fact that $(1+\frac{\epsilon}{2})\delta_k(t)\leq\frac{\epsilon}{2}$ from
\eqref{eq:delta_bound}. Due to the general formulation of solution \eqref{eq:FPEsol}, there are
\begin{align}
\rho _{k+1}(t) =\int ^{t}_{0}\rho_{k}(s) q_{k}(s) e^{-\int ^{t}_{s}\qhat_{k+1}(u) \diff u}\diff s \\
\rhohat _{k+1}(t) =\int ^{t}_{0}\rhohat_{k}(s)\qhat_{k}(s) e^{-\int ^{t}_{s}\qhat_{k+1}(u) \diff u}\diff s
\end{align}
Then we can bound their difference as follows,
\begin{align*}
\frac{|\rhohat_{k+1}(t)-\rho_{k+1}(t)|}{\rho_{k+1}(t)}&\leq\frac{1}{\rho_{k+1}(t)}\int_0^t
|\rho_k(s)q_k(s)-\rhohat_k(s)\qhat_k(s)|e^{-\int ^{t}_{s}\qhat_{k+1}(u) \diff u}\diff s \\
&\leq\frac{\epsilon}{\rho_{k+1}(t)}\int ^{t}_{0}\rho_{k}(s) q_{k}(s) e^{-\int ^{t}_{s}\qhat_{k+1}(u) \diff u}\diff s =\epsilon.
\end{align*}
For $j=k+1,\dots,K$, $\qhat_{j}(t)$ is the same for both $\rho(t)$ and $\rhohat(t)$ in \eqref{eq:FPEsol}, one can readily check that
$|\rhohat_j(t)-\rho_ j(t)|/\rho_j(t)\leq \epsilon$ implies that $|\rhohat_{j+1}(t)-\rho_{j+1}(t)|/\rho_{j+1}(t)\leq \epsilon$,
Therefore,
\begin{equation}
\frac{|\influhat(t)-\influ(t)|}{\influ(t)} \leq \frac{1}{\influ(t)}\sum_{j=k}^Kj|\rhohat_j(t)-\rho_j(t)|
\leq \frac{\epsilon}{\influ(t)}\sum_{j=k}^Kj\rho_j(t)\leq \epsilon
\end{equation}
for all $t\geq0$, 
which completes the proof.
\end{proof}

\begin{theorem}\label{thm:influ_error}
Let $\epsilon\in(0,1)$, and $\rho(t)$ and $\rhohat(t)$ be the solutions of $\rho'(t)=\rho(t) Q(t)$
and $\rhohat'(t)=\rhohat(t) \Qhat(t)$ where $\Qhat(t):=Q_{0}(t)$, 
respectively, and $\influ(t)=\sum_{k=0}^{K} k\rho_k(t)$
and $\influhat(t)=\sum_{k=0}^{K} k\rhohat_k(t)$.
If \eqref{eq:delta_bound} holds for $k=0,\dots,K-1$ and there exist upper bound $\alphabar$ and
lower bound $\alphaul>0$ for all activation rates in $G=(V,E)$, then
\begin{equation}\label{eq:influ_error}
\frac{|\influhat(t)-\influ(t)|}{\influ(t)}\leq [(1+\epsilon)^K-1]\min\left\{1,c_K(t)e^{-\alphaul t}\right\}, \quad \forall t\geq0,
\end{equation}
where $\qbar:=\max_{k}\{q_k\}$ is bounded and $c_K(t):=\frac{1}{K}\sum_{j=0}^{K-1}\frac{K-j}{j!}(\qbar t)^j$.
\end{theorem}

\begin{proof}
For every $k=0,1,\dots,K-1$, let $\influ_k(t)$ be the influence estimated by solving 
differential equation $\rho'(t)=\rho(t) Q_k(t)$.
Then Lemma \ref{lemma:rho_error} shows that $|\influ_k(t)-\influ_{k+1}(t)|/\influ_{k+1}(t)\leq \epsilon$
for $k=0,\dots,K-2$ and $|\influ_{K-1}-\influ(t)|/\influ(t)\leq \epsilon$ 
provided \eqref{eq:delta_bound} holds for all $k$. Therefore
$1-\epsilon\leq \frac{\influ_k(t)}{\influ_{k+1}(t)}\leq 1+\epsilon$ and
$1-\epsilon\leq \frac{\influ_{K-1}(t)}{\influ(t)}\leq 1+\epsilon$, and hence
\begin{equation}\label{eq:influ_multi}
(1-\epsilon)^K \leq \frac{\influhat(t)}{\influ(t)}=\frac{\influ_{0}(t)}{\influ(t)}=\frac{\influ_{K-1}(t)}{\influ(t)}\cdots
\frac{\influ_1(t)}{\influ_2(t)}\frac{\influ_0(t)}{\influ_1(t)}\leq (1+\epsilon)^K.
\end{equation}
Therefore $|\influhat(t)-\influ(t)|/\influ(t)\leq \max\{1-(1-\epsilon)^K,(1+\epsilon)^K-1\}=(1+\epsilon)^K-1$.

On the other hand, we have $\alphaul \leq q_k(t) \leq \alphabar k \min\{\dbar,K-d\}$
and hence $\alphaul\leq q_k(t) \leq \qbar$ for all $k=0,1,\dots,K-1$ and $t\geq0$.
Here $\qbar\leq \alphabar \dbar(K-\dbar)$ if $\dbar\leq \frac{K}{2}$ 
and $\qbar\leq \frac{\alphabar K^2}{4}$ if $\dbar>\frac{K}{2}$.
By induction we claim that $\rho_k(t)\leq \frac{(\qbar t)^k}{k!}e^{-\alphaul t}$ for $k=0,\dots,K-1$ as follows:
the claim is obviously true for $k=0$; suppose it is true for $k\leq K-2$, then
\begin{align}
\rho _{k+1}\left( t\right) &=\int ^{t}_{0}\rho _{k}\left( s\right) q_{k}\left( s\right) e^{-\int ^{t}_{s}q_{k+1}\left( u\right) \diff u}\diff s 
\leq \int_0^t \frac{(\qbar s)^k}{k!}e^{-\alphaul s} \qbar e^{-\alphaul (t-s)} \diff s \nonumber\\
& =\frac{\qbar^{k+1}e^{-\alphaul t}}{k!}\int_0^t s^k \diff s =
\frac{(\qbar t)^{k+1}}{(k+1)!}e^{-\alphaul t}. \nonumber
\end{align}
Moreover, from Lemma \eqref{lemma:rho_error} 
we can readily deduce that $(1-\epsilon)^{j+1}\leq \rhohat_j(t)/\rho_j(t) \leq (1+\epsilon)^{j+1}$
similar as for \eqref{eq:influ_multi}.
Hence $|\rhohat_j(t)-\rho_j(t)|/\rho_j(t)\leq \epsilon_j:=(1+\epsilon)^{j+1}-1$
for $j=0,1,\dots,K-1$. Therefore, we have
\begin{align*}
\frac{|\influhat(t)-\influ(t)|}{\influ(t)} &= 
\frac{1}{\influ(t)}\left|\sum_{j=0}^K j\left(\rhohat_j(t)-\rho_j(t)\right)\right|
=\frac{1}{\influ(t)}\left|\sum_{j=0}^{K-1} (K-j)\left(\rhohat_j(t)-\rho_j(t)\right)\right| \\
&\leq \frac{1}{\influ(t)}\sum_{j=0}^{K-1} (K-j)\left|\rhohat_j(t)-\rho_j(t)\right| 
\leq \frac{1}{\influ(t)}\sum_{j=0}^{K-1} (K-j)\epsilon_j\rho_j(t) \\
&\leq \frac{\epsilon_{K-1}}{|S|}\sum_{j=0}^{K-1} (K-j)\rho_j(t) 
\leq \frac{\epsilon_{K-1}e^{-\alphaul t}}{|S|}\sum_{j=0}^{K-1} \frac{K-j}{j!}(\qbar t)^j\\
&=\epsilon_{K-1} c_K(t)e^{-\alphaul t}
\end{align*}
where we used the fact that $\rho_K(t)=1-\sum_{j=0}^{K-1}\rho_j(t)$
in the second equality, $\epsilon_{K-1}\geq \epsilon_j$ for all $j=0,\dots,K-1$ and 
$\mu(t)\geq |S|$ in the fourth inequality\footnote{The lower bound $\mu(t)\geq |S|$
is loose as $\mu(t)$ increases from $|S|$ to $K$ along $t$. This is not an issue in the estimate above
if $|S|\geq 1$. If $|S|=0$ then one can assume existence of a pre-activated node (in addition to $V$) that
activates each $i\in V$ at rate $\beta_i$ since $t=0$ to mimic the self-activations, 
and a modified estimate can be applied trivially so we omit the details here.}.
Combining the two bounds of $|\influhat(t)-\influ(t)|/\influ(t)$ above, we obtain \eqref{eq:influ_error}.
\end{proof}

Theorem \ref{thm:influ_error} shows that an $O(1/t)$ decay of error in estimated $\qhat_k(t)$ 
results in an exponential $O(e^{-\alphaul t})$ decay  of error in predicted influence $\influhat(t)$.
This result implies that for an exponentially decaying error in $\influhat(t)$
the estimation error in $\qhat_k(t)$ only needs to remain about as constant
for all sufficiently large $t$.

\begin{corollary}\label{cor:influ_error}
Suppose $\rho(t),\rhohat(t),\influ(t),\influhat(t)$ are defined and conditions for 
$\alphabar$ and $\alphaul$ hold as in Theorem \ref{thm:influ_error}.
Let $\varepsilon>0$ and $c\in(0,\alphaul)$, then
$|\influhat(t)-\influ(t)|/\influ(t)\leq \varepsilon e^{-ct}$ as long as the estimated $\qhat_k(t)$ satisfies
\begin{equation}\label{eq:q_bound}
\frac{|\qhat_k(t)-q_k(t)|}{q_k(t)}\leq 
\frac{\alphaul-c}{K \qbar_k}+\frac{\log\varepsilon - K\log2-\log c_K(t)}{K\qbar_k t}=C_k-O\left(\frac{\log t}{t}\right)
\end{equation}
for each $k=0,1,\dots,K-1$, where $\qbar_k:=\alphabar k \min\{\dbar,K-k\}$ and $C_k:=(\alphaul-c)/K\qbar_k$.
\end{corollary}
\begin{proof}
By Theorem \ref{thm:influ_error} and the bound of error $\delta_k(t)$ in \eqref{eq:delta_bound}, 
we can attain $|\influhat(t)-\influ(t)|/\influ(t)\leq \varepsilon e^{-ct}$ as long as
$\delta_k(t)$ satisfies $\qbar_k t \delta_k(t) \leq \log(1+\epsilon(t))$ for some $\epsilon(t)$
such that $[(1+2\epsilon(t))^K-1]c_K(t)e^{-\alphaul t}= \varepsilon e^{-ct}$. To this end, we need
$\log(2e^{\qbar_k t \delta_k(t)}-1)\leq\frac{1}{K}\log(\frac{\varepsilon e^{(\alphaul-c)t}}{c_K(t)}+1)$,
to guarantee which it suffices to have $ \log(2e^{\qbar_k t \delta_k(t)}) \leq \frac{1}{K}\log(\frac{\varepsilon e^{(\alphaul-c)t}}{c_K(t)})$,
i.e., \eqref{eq:q_bound}.
\end{proof}

\section{Experimental results} 
\label{sec:results}
We first apply the proposed method to networks (with various sizes and parameters) 
generated by four 
models commonly used in social/biological/contact networking applications:
Erd\H{o}s-R\'{e}nyi's random, small-world, scale-free, and 
Kronecker network\footnote{Code for generating 
Kronecker network is at \url{https://github.com/snap-stanford/snap/tree/master/examples/krongen}
and other three using CONTEST package at 
\url{http://www.mathstat.strath.ac.uk/outreach/contest/toolbox.html}}.
Activation rates $\{\alpha_{ij}\}$ are drawn from interval $(0,1)$ uniformly
to simulate the inhomogeneous propagation rates across edges.
Unless otherwise noted, we only consider node-to-node activations
in propagations without self-activation and recovery.
In all cases except those in Fig.\ \ref{fig:tree}, exact solutions for influence
are computationally infeasible due to the
large size and heterogeneous transmission rates between nodes,
we therefore use enough Monte Carlo Markov chain (MCMC) simulated cascades (5000 cascades for each network)
to compute the ground truth density $\rho(t)$ and influence $\influ(t)$.

In Fig.\,\ref{fig:tree}, we show the performance 
of our method based on Fokker-Planck equation in Section \ref{subsec:rates} 
using \texttt{FPE-dist} and \texttt{FPE-tree}. The NIMFA (N-interwined 
mean field approximation) is a state-of-the-art method that uses mean-field theory to
obtain a system of differential equations to calculate the probability $p_i(t)$ (node $i$ gets activated at time $t$)
\cite{Van-Mieghem:2013a,Van-Mieghem:2009a}, and estimates the influence by $\sum_i p_i(t)$.
Note that we take a completely different approach to calculate the probability $\rho_k(t)$
for each possible influence size $k$ and estimate the influence by $\sum_k k\rho_k(t)$.
For the influence prediction test, we find that our approach appears to be more accurate 
as shown in Fig.\,\ref{fig:tree}, especially \texttt{FPE-tree} which matches ground truth (MCMC)
very closely (but at the expense of higher computational cost to estimate transition rates $q_k(t)$). 
The \texttt{FPE-dist} also provides reasonably accurate solution but requires much lower computational cost, hence we only
use this version in other tests with large networks.
Note that NIMFA requires solving a nonlinear system of $K$ differential equations numerically and hence 
has the same order of computation complexity as our approach.
\begin{figure}
\centering
\includegraphics[width=.28\textwidth]{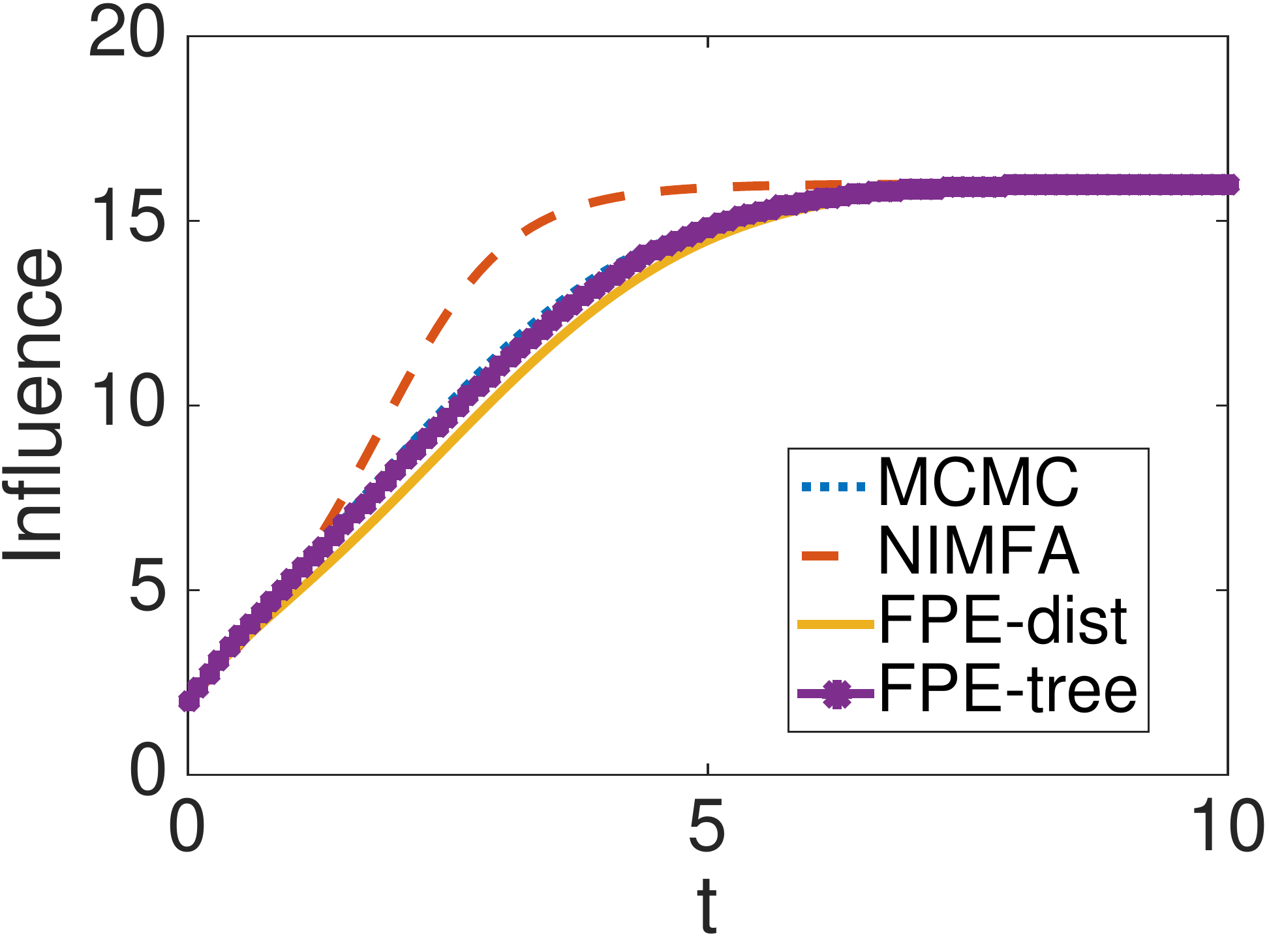}
\includegraphics[width=.28\textwidth]{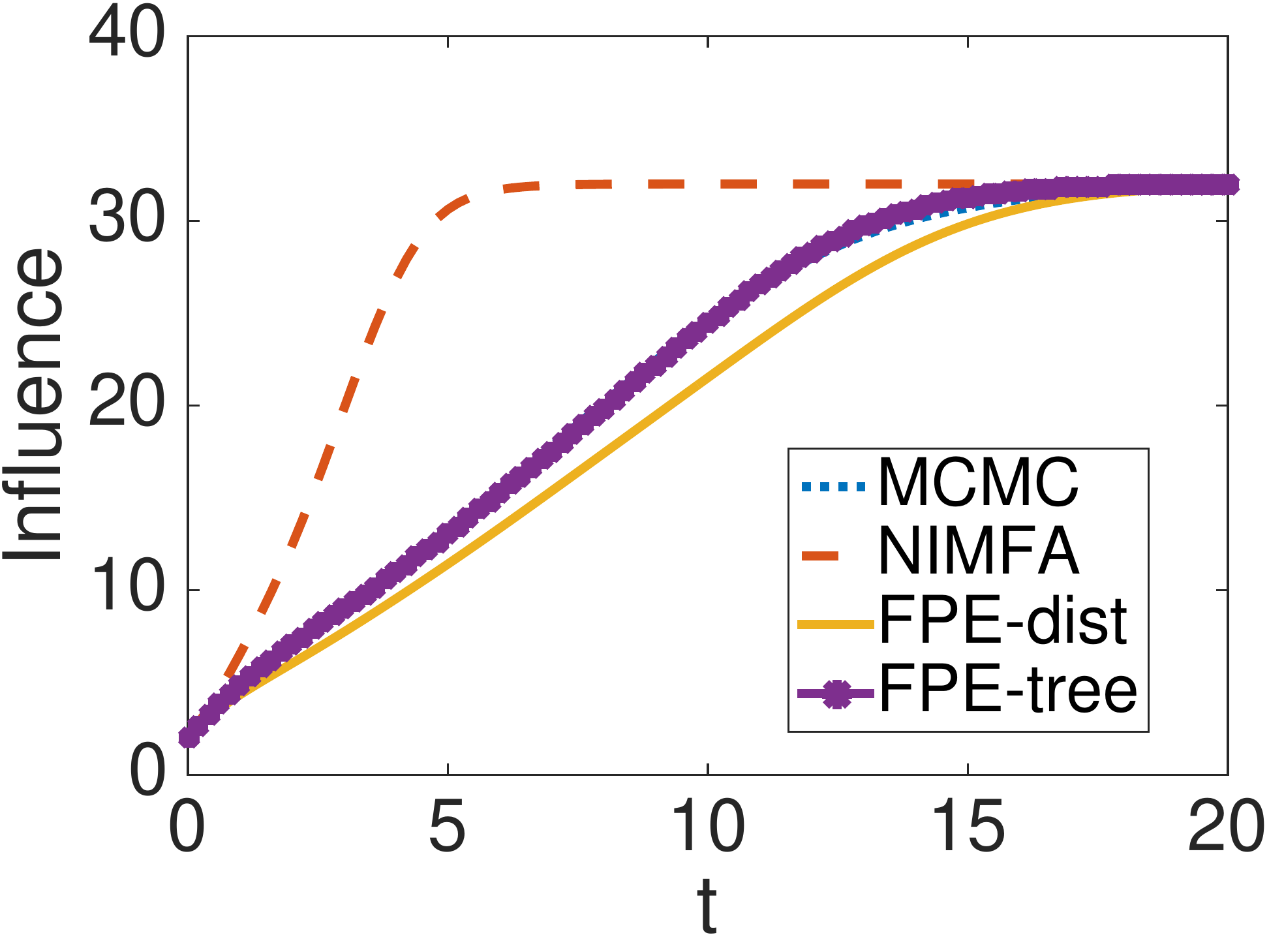}
\includegraphics[width=.28\textwidth]{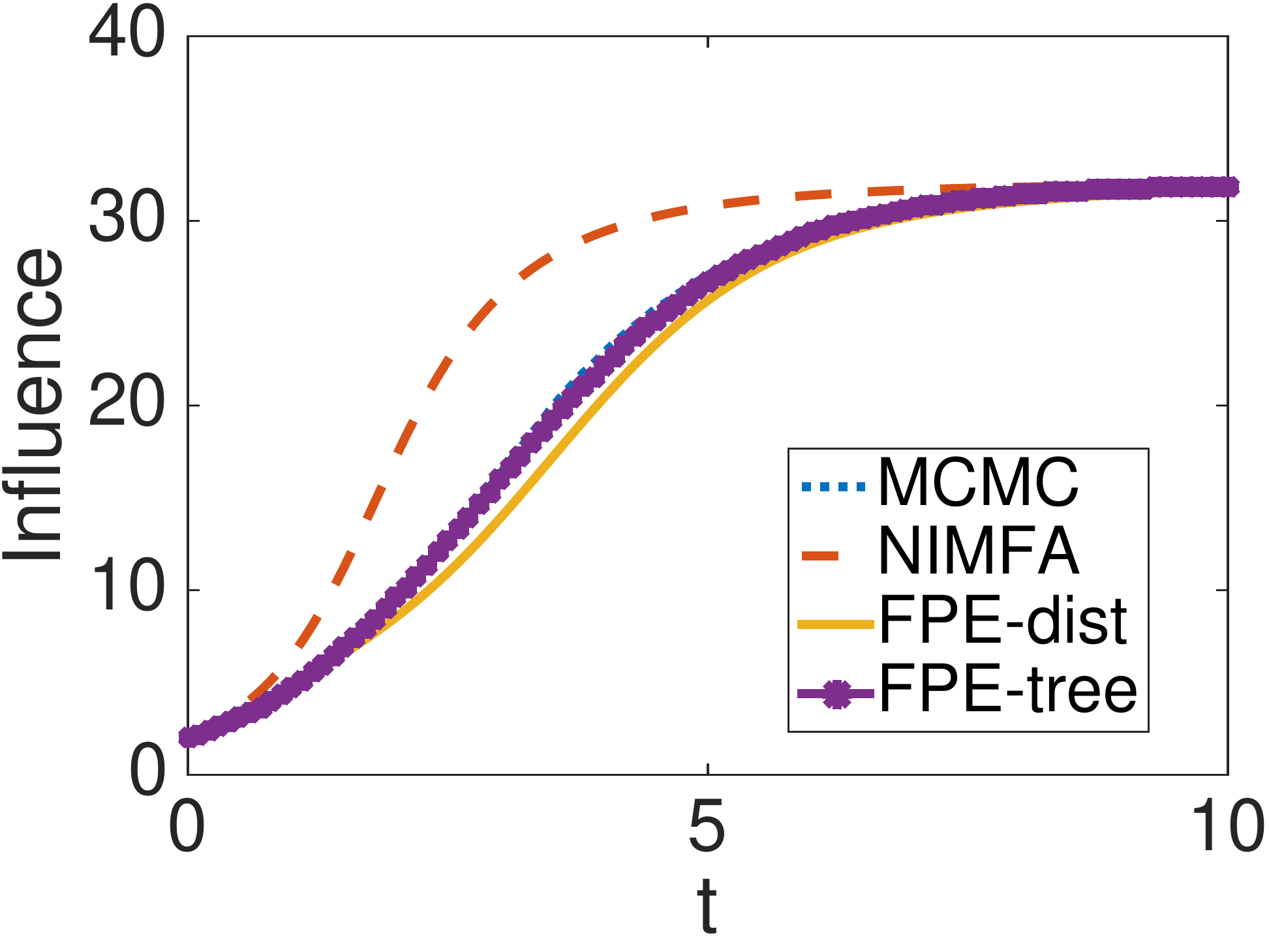}
\caption{Influence prediction on small sized network (when our matlab implementation 
of \texttt{FPE-dist} still takes short time in computing). 
\textbf{Left two}: Erd\H{o}s-R\'{e}nyi's network of size $K=16,32$. \textbf{Right}: Small-world network $K=32$.
Average degree $(1/K)\sum_i|N_i^{\mathrm{out}}|=4$.}
\label{fig:tree}
\end{figure}

In Fig.\,\ref{fig:varnet}, we show the influence prediction result on networks of much larger
size $K=1024$. Despite of very different network structures, \texttt{FPE-dist} provides faithful
influence prediction and matches ground truth (MCMC) closely.
\begin{figure}
\centering
\includegraphics[width=.28\textwidth]{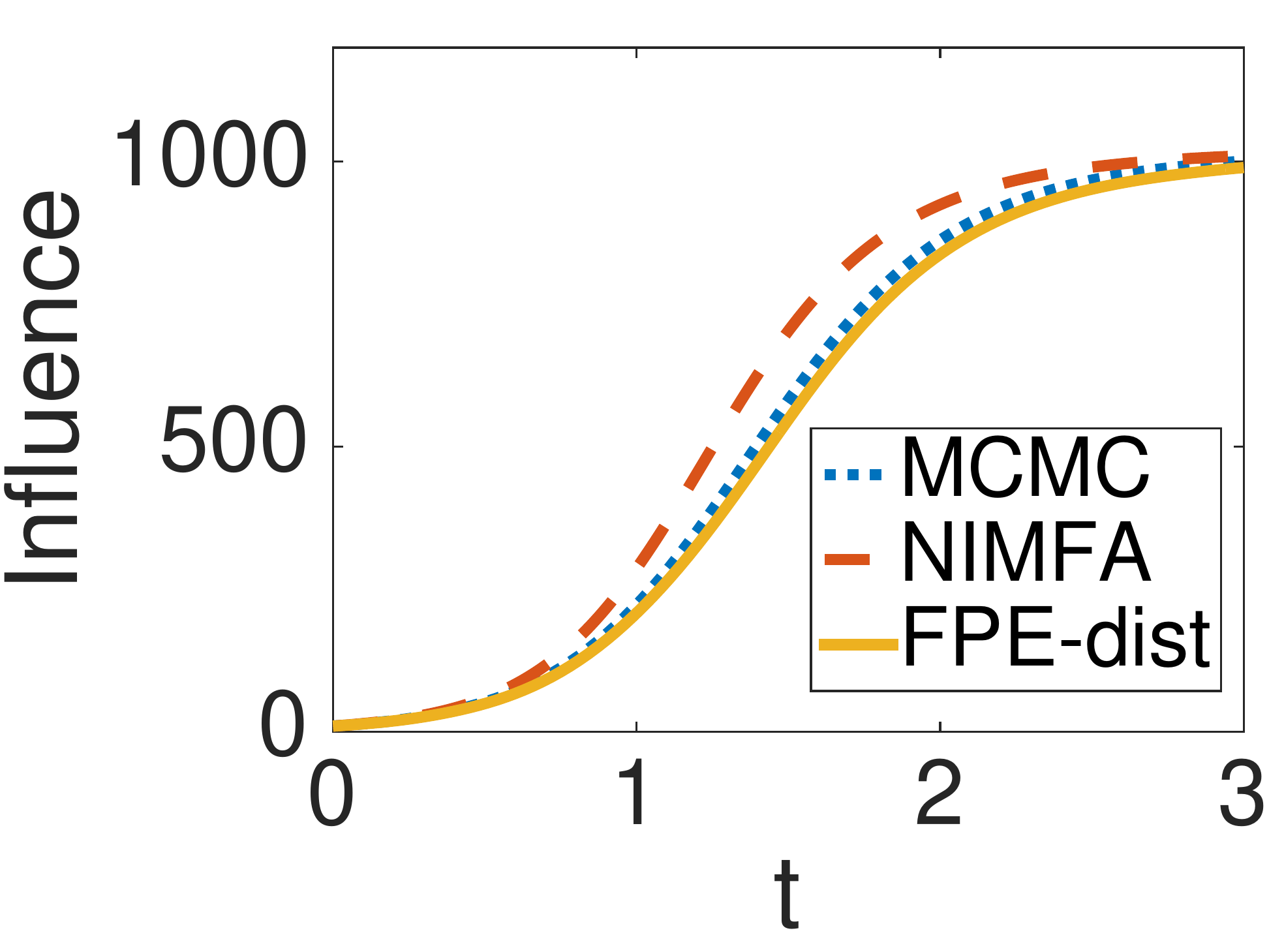}
\includegraphics[width=.28\textwidth]{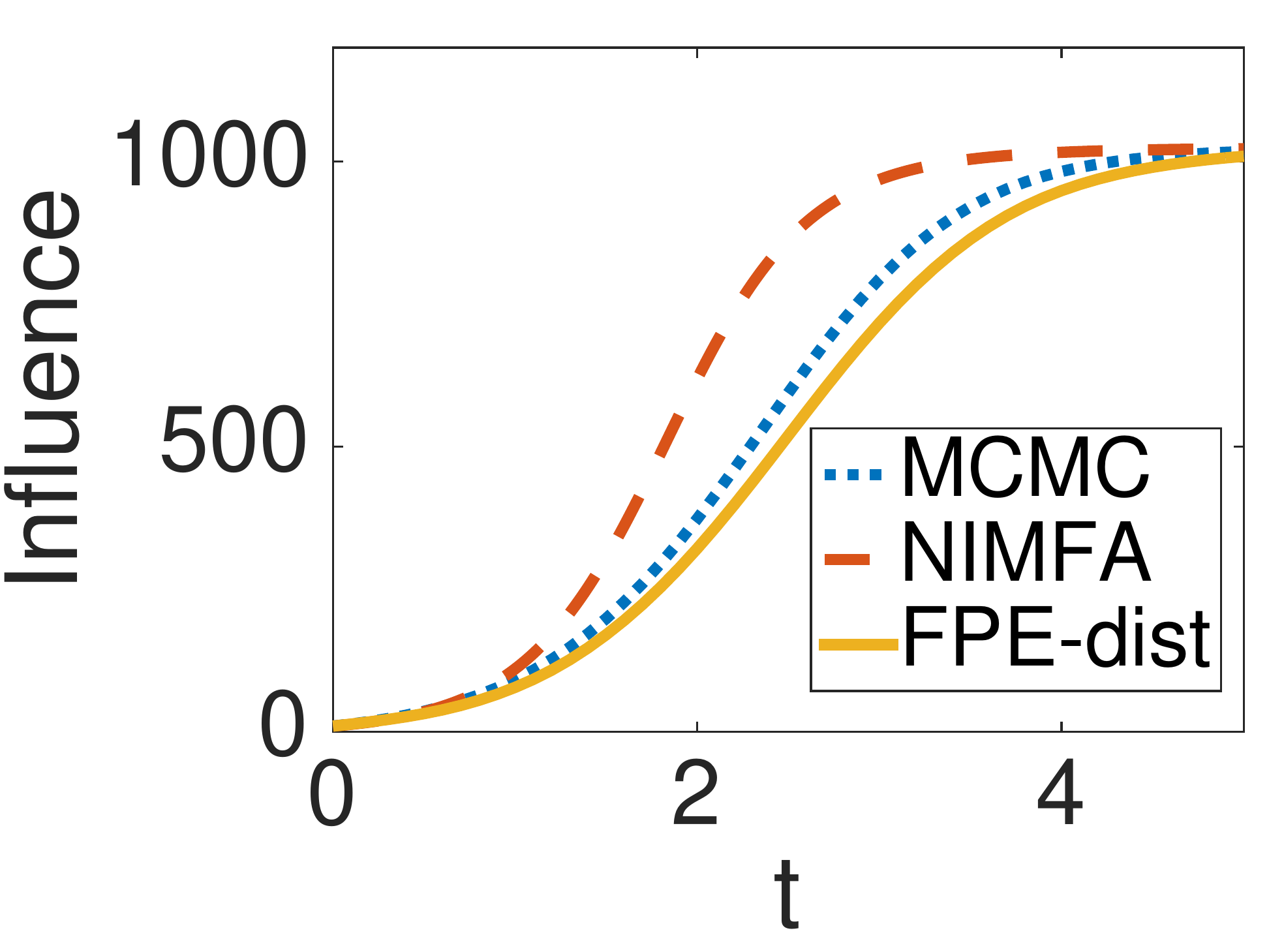}
\includegraphics[width=.28\textwidth]{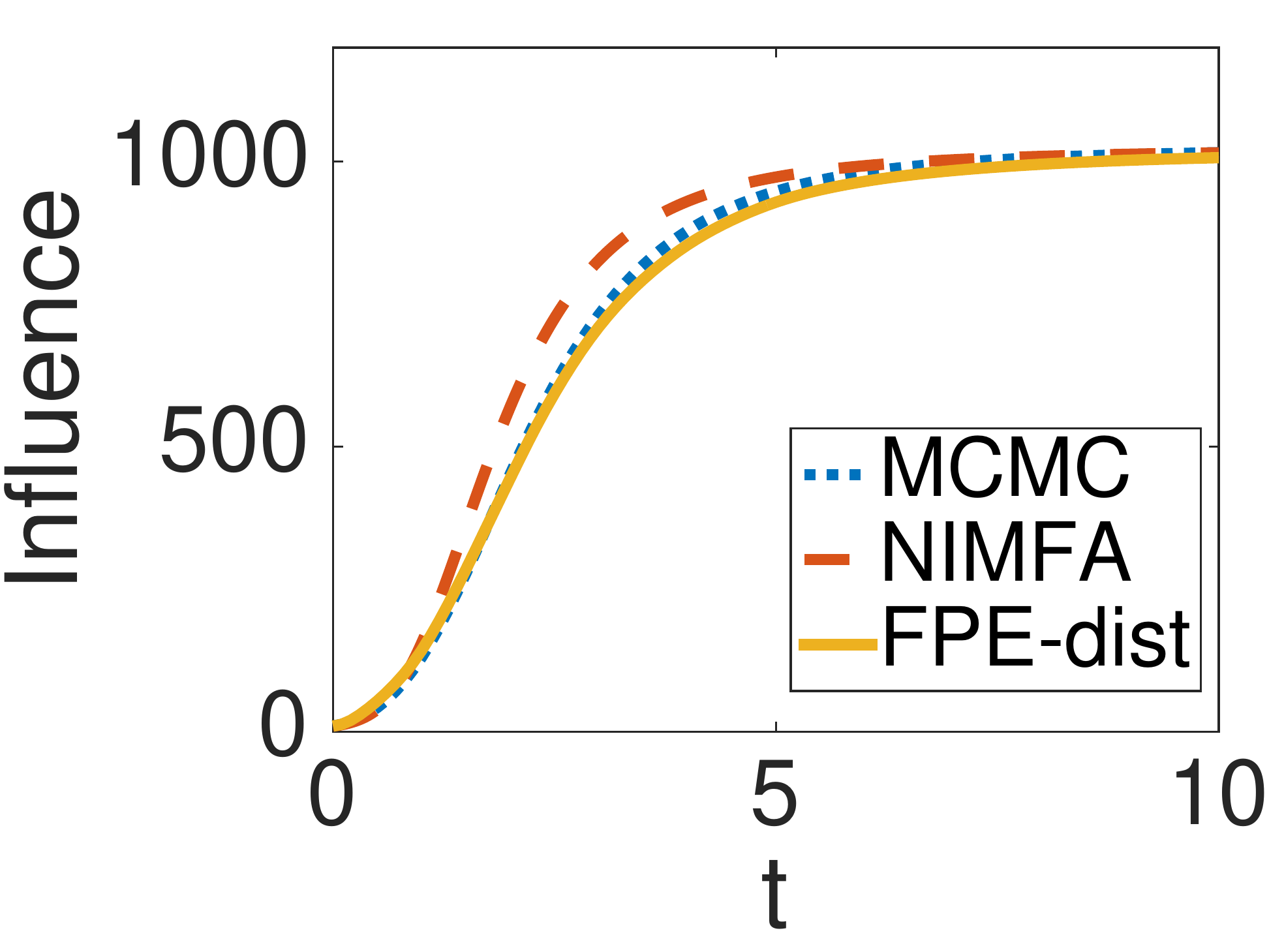}
\caption{Influence prediction on \textbf{Left}: Erd\H{o}s-R\'{e}nyi's network,
\textbf{Middle}: small-world network, 
and \textbf{Right}: scale-free network. All have size $K=1024$ and 
average degree are $(1/K)\sum_i|N_i^{\mathrm{out}}|=8,6,6$ respectively.}
\label{fig:varnet}
\end{figure}

Influence prediction problem is considered very challenging computationally,
especially for dense networks. In Fig.\,\ref{fig:robust} we test \texttt{FPE-dist}
on very dense Erd\H{o}s-R\'{e}nyi's random networks of size $K=1024$
where average degrees are $(1/K)\sum_i|N_i^{\mathrm{out}}|=32$, $64$, and $128$ respectively.
On all of these networks, \texttt{FPE-dist} returns highly accurate prediction of influence
which justifies its robustness.

The influence prediction problem considered in this paper, as noted in Section \ref{subsec:problem_description}, 
is significantly different from those for dynamical processes on
on networks in statistical physics. Our network is deterministically heterogenous, 
meaning that $G=(V,E)$ and $\alpha_{ij}$ on all edges are 
given, and they play critical roles in propagations. Therefore, the identities of nodes
in source set $S$ matter greatly (in contrary the nodes in a network are not distinguishable
in most statistical physics problems)
which leads to many important follow-up questions such as
influence maximization (e.g., finding the source set $S$ that solves 
$\max_{|S|\leq k_0}\influ(t;S)$ for some prescribed size $k_0\in\mathbb{N}$ and
time $t$)
\cite{Cohen:2014a,Gomez-Rodriguez:2012c,Kempe:2003a,Wang:2012a}
and outbreak detection \cite{Cui:2013a,Leskovec:2007b}. To see the critical role of source set $S$, 
we apply \texttt{FPE-dist} to three different choices of source set $S_1,S_2,S_3$ all with $|S_i|=10$
and show the prediction results in the middle panel of Fig.\,\ref{fig:robust}.
Here $S_1$ is the choice obtained by the influence maximization function 
from ConTinEst code \cite{Du:2013a}, $S_2$ consists of the ten nodes with largest degrees in $G$,
and $S_3$ contains ten nodes randomly chosen from the network.
The plots clearly show different influences of these sources sets $S_i$'s due to the
deterministically heterogeneous structure of the network.
Nevertheless, \texttt{FPE-dist} has very robust performance and matches 
the ground truths (MCMC) closely in every case. 

We also compare \texttt{FPE-dist} to the state-of-the-arts learning-based 
ConTinEst algorithm \cite{Du:2013a}. 
The network data and its implementation are obtained from the ConTinEst package published by 
its authors\footnote{Data and code available at \url{http://www.cc.gatech.edu/~ndu8/DuSonZhaMan-NIPS-2013.html}.}.
ConTinEst is a state-of-the-arts learning-based algorithm that uses parametrized kernel
functions to approximate the coverage of each node based on 
Monte Carlo samplings. 
The result is shown in the right panel of Fig.\,\ref{fig:robust}.
From this test, we see that \texttt{FPE-dist} is very accurate as it matches the
ground truth (MCMC) much better. Moreover, ConTinEst takes excessively long
time to estimate influence for denser networks as those in the left panel of 
Fig.\,\ref{fig:robust}, while \texttt{FPE-dist} still works robustly without suffering the issue at all.
Note that comprehensive comparison of 
ConTinEst with several other existing methods is reported in \cite{Du:2013a}, 
from which significant improvement in accuracy of the proposed method \texttt{FPE-dist}
can be projected.
\begin{figure}
\centering
\includegraphics[width=.28\textwidth]{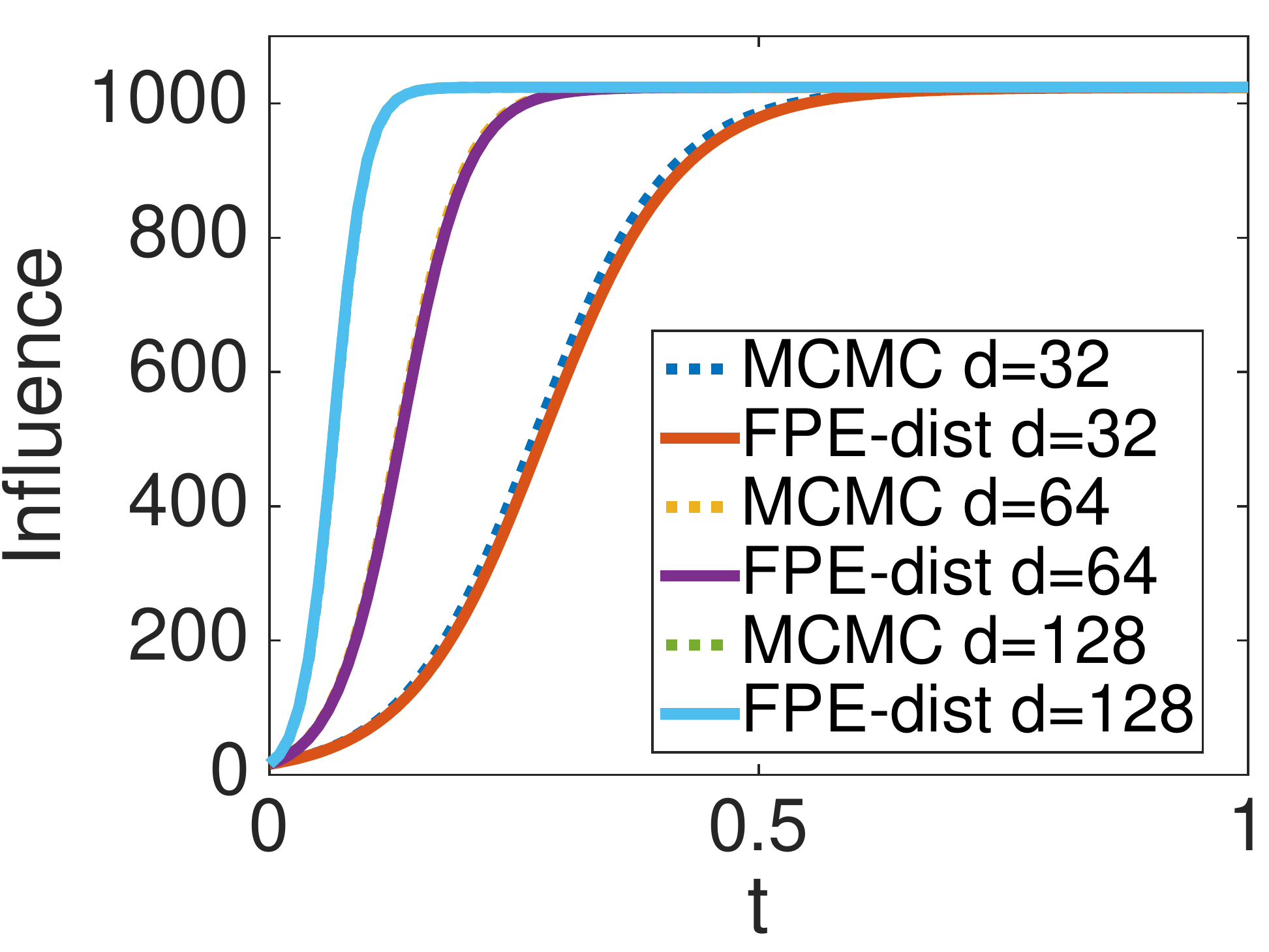}
\includegraphics[width=.28\textwidth]{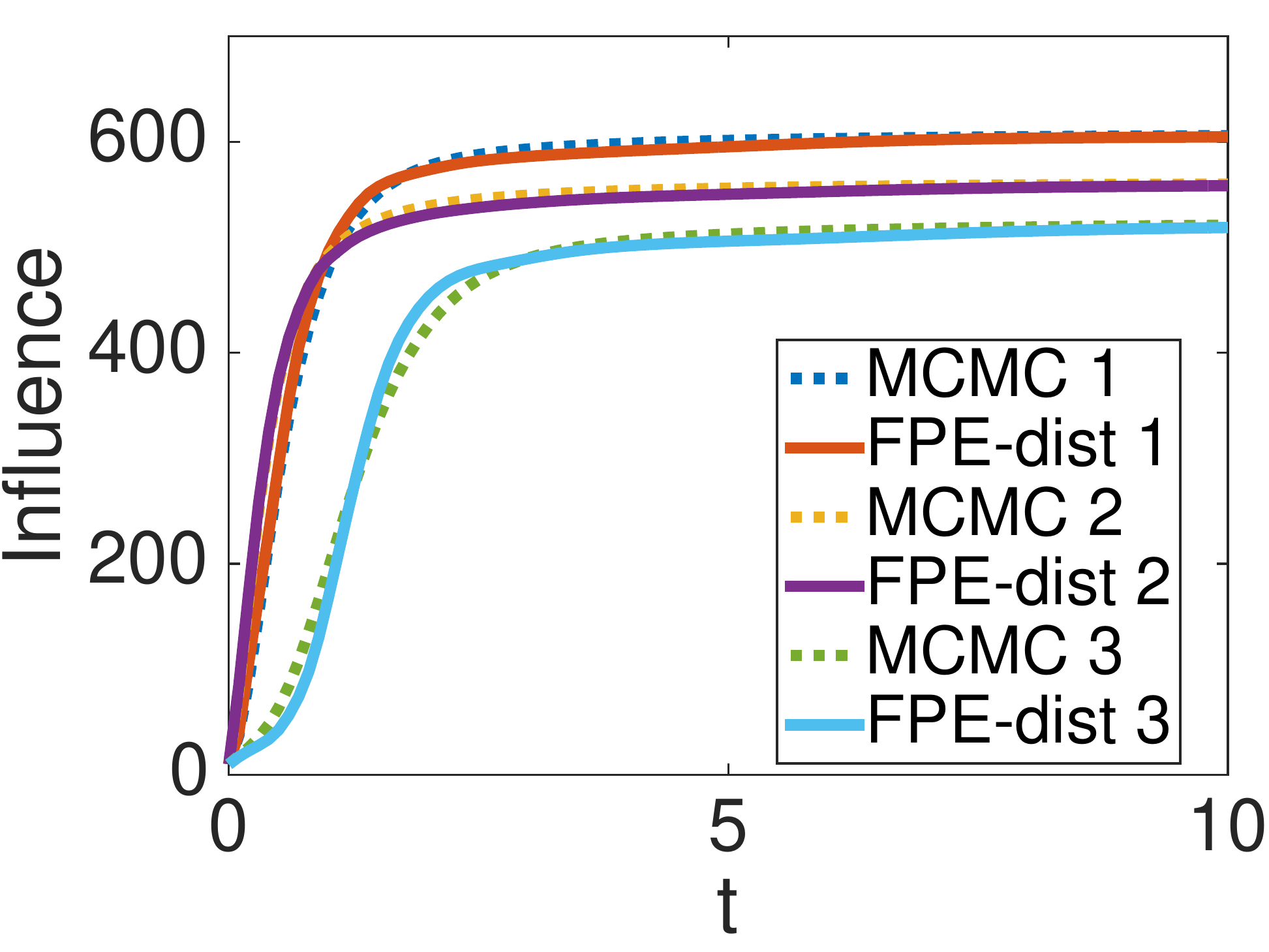}
\includegraphics[width=.28\textwidth]{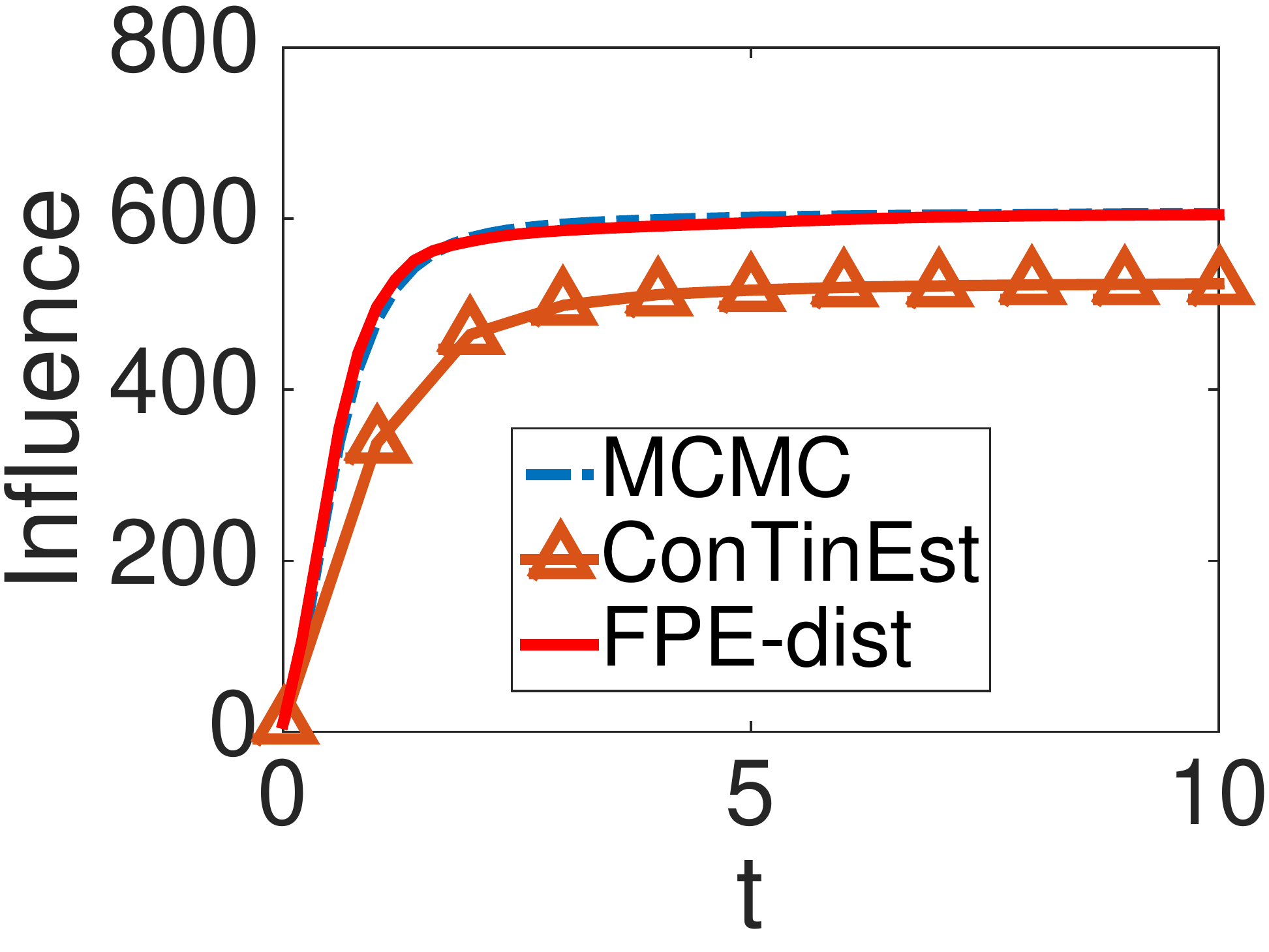}
\caption{\textbf{Left}: Influence prediction on dense Erd\H{o}s-R\'{e}nyi's random network
with $K=1024$ and $(1/K)\sum_i|N_i^{\mathrm{out}}|=32,64,128$ respectively.
\textbf{Middle}: Influence prediction on the \underline{same} Kronecker network of size $1024$ 
using three different choices of source set $S_1,S_2,S_3$ ($|S_i|=10$ in all three cases).
\textbf{Right}: Comparison with the state-of-the-arts learning-based ConTinEst method.}
\label{fig:robust}
\end{figure}

We established the relation between estimation error in $\{q_k(t)\}$ and 
the prediction error in $\influ(t)$ in Section \ref{subsec:error}.
To check this numerically, we apply \texttt{FPE-dist} to a dense Erd\H{o}s-R\'{e}nyi's network
of size $K=300$ and average degree $(1/K)\sum_i|N_i^{\mathrm{out}}|=150$ ($\alpha_{ij}$
again drawn from $(0,1)$ uniformly) with source set $S=\{1,\dots,10\}$, 
and check the estimated $q_k(t)$, $\rho_k(t)$
and $\influ(t)$ with those obtained by MCMC simulated cascades. 
Recall that \texttt{FPE-dist} uses very crude estimate of $q_k(t)$ 
by setting a constant $\qhat_k=\alpha(U_k^*)$ where $U_k^*$ 
contains the $k$ nodes of shortest distance from $S$ in Section \ref{subsec:rates}.
We first plot the $\qhat_k$ for $k=10,70,130,190$ and compare with
$q_k(t)$ given by ground truth (MCMC simulations) in the top row of
Fig.\,\ref{fig:cputime}. To this end, we observe from \eqref{eq:FPEform} that
$q_k(t)=-(\sum_{j=0}^k \rho_j(t))'/\rho_k(t)$, so we obtain $\rho_k(t)$ from MCMC simulations and 
apply finite difference to get $\rho_k'(t)$ and hence $q_k(t)$. 
Note that $q_k(t)=0$ for most $t$ because $\sum_{j=0}^k \rho'(t)$ or $\rho_k(t)$ vanish there and obtaining 
these $q_k(t)$ is unstable numerically, so the comparison is only meaningful
for $t$ where $\rho_k(t)$ is away from zero. From top row of Fig.\,\ref{fig:cputime}, we can see 
the estimated $\qhat_k$ appear to accurately captured the mean of $q_k(t)$ 
but can be quite deviated (i.e., with large $|\qhat_k(t)-q_k(t)|/q_k(t)$). However,
the densities $\rhohat_k(t)$ computed using these $\qhat_k$ are still close
to the ground truth $\rho_k(t)$, as shown by the small relative error $|\rhohat_k(t)-\rho_k(t)|/\rho_k(t)$
in the bottom leftmost panel of Fig.\,\ref{fig:cputime}.
This also yields a small relative error in influence prediction $|\influhat(t)-\influ(t)|/\influ(t)$
(second on bottom row), and close match of prediction result $\influhat(t)$ and
ground truth $\influ(t)$ (MCMC) (third on bottom row) in Fig.\,\ref{fig:cputime}.
The small errors in $\rhohat_k(t)$ and $\influhat(t)$ in our numerical tests suggest that the 
theoretical bound on the estimation error in $q_k(t)$ in \eqref{eq:delta_bound} may be further relaxed without
degrading solution quality.
\begin{figure}
\centering
\includegraphics[width=.24\textwidth]{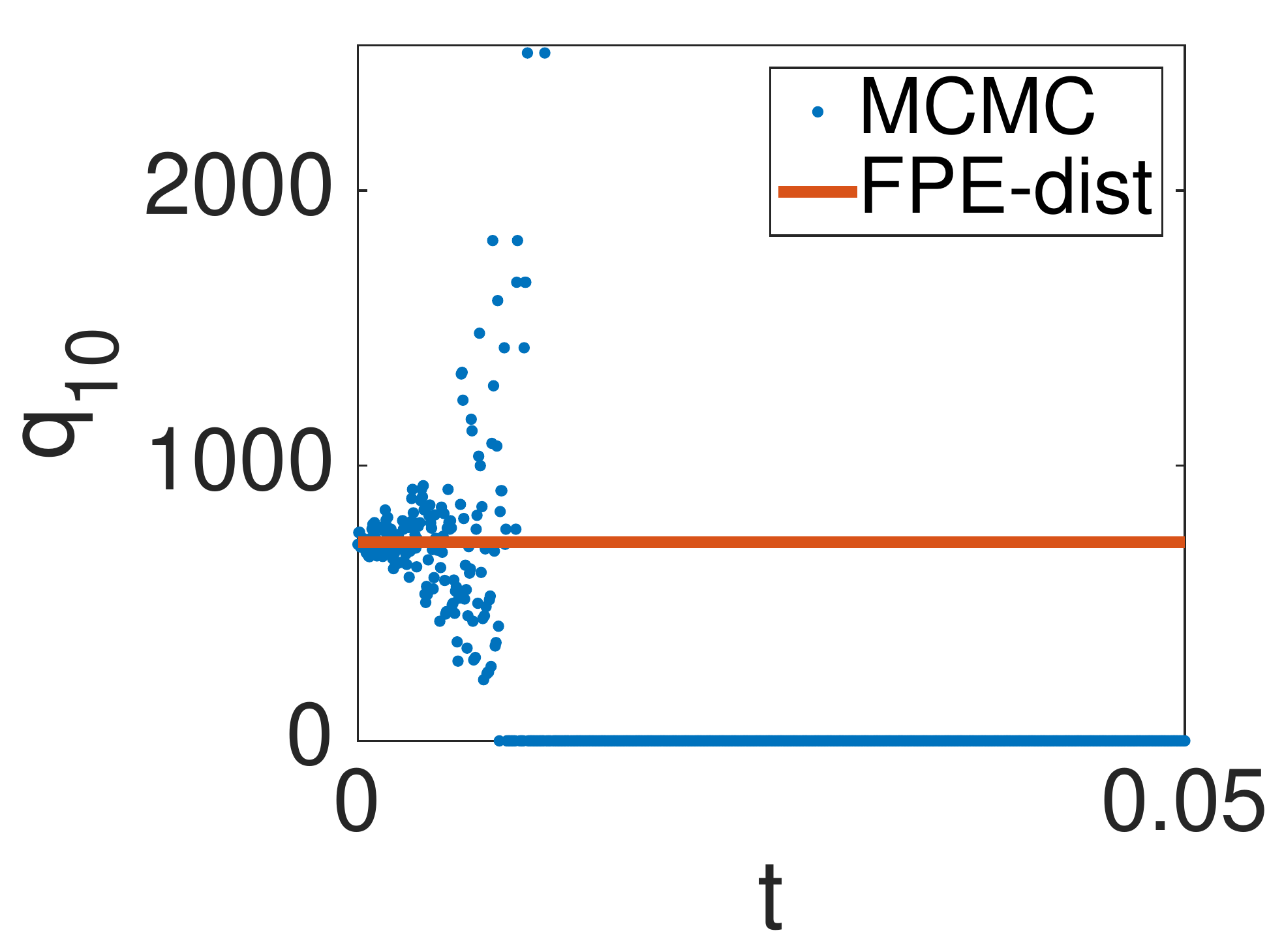}
\includegraphics[width=.24\textwidth]{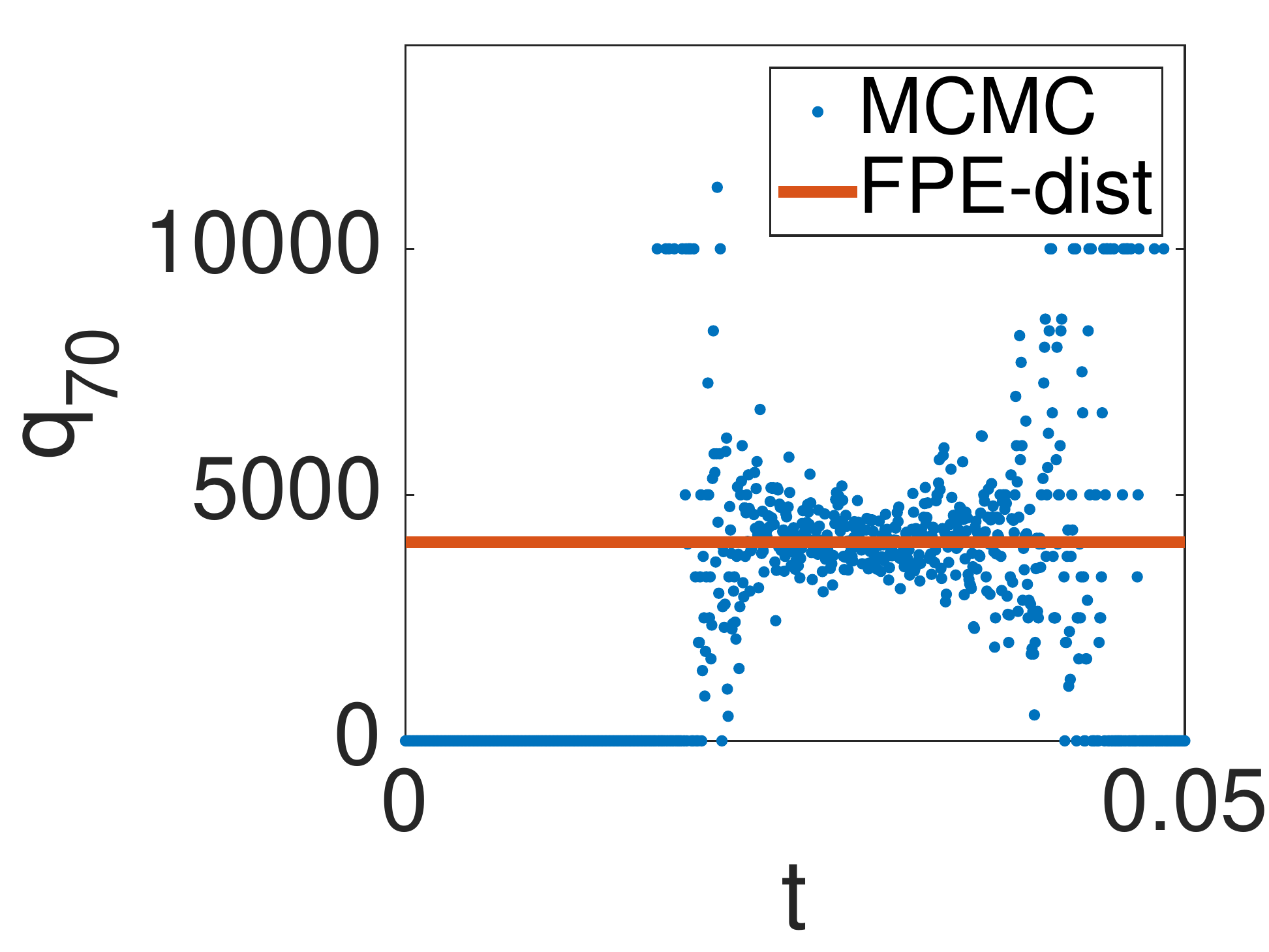}
\includegraphics[width=.24\textwidth]{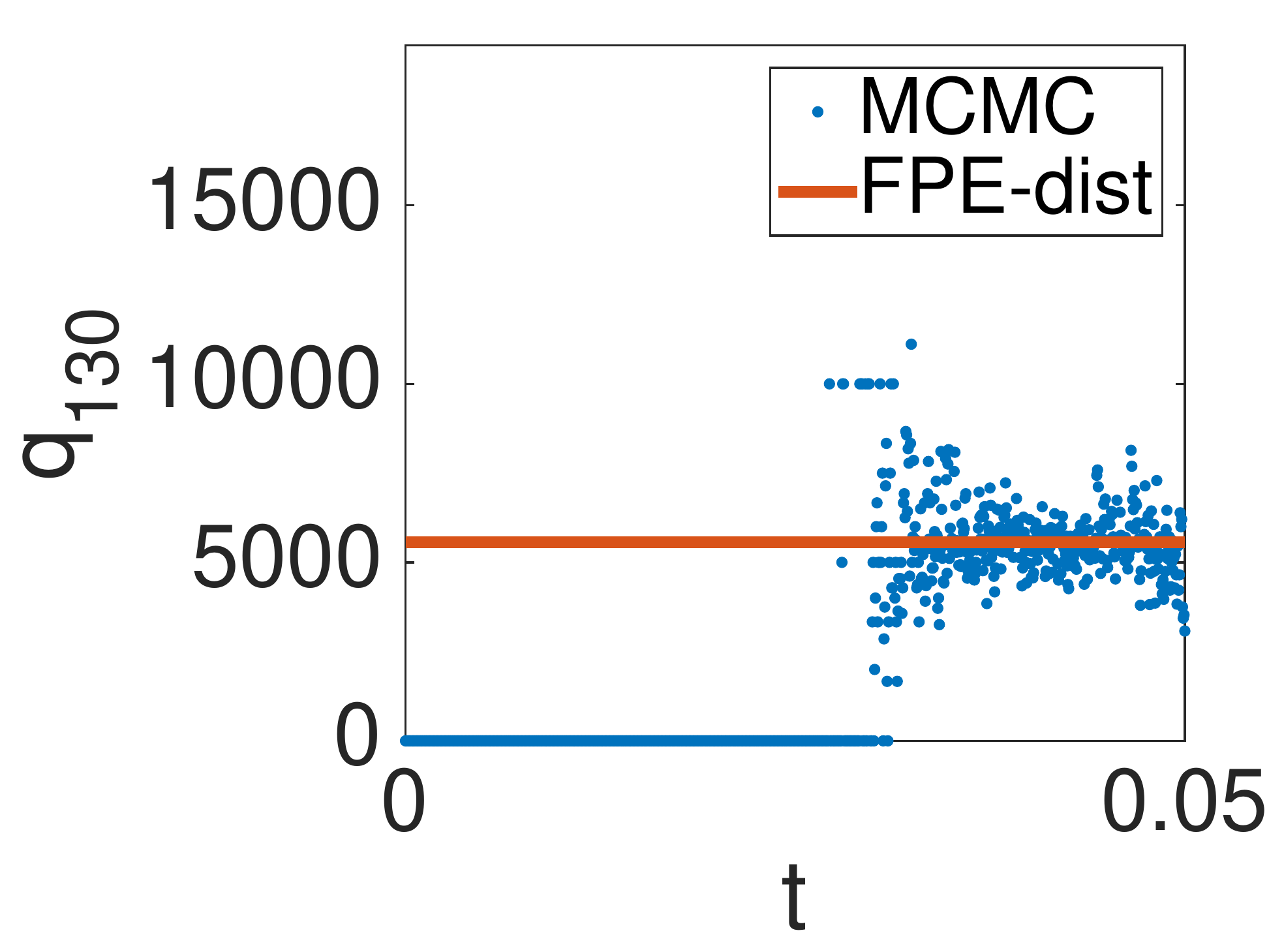}
\includegraphics[width=.24\textwidth]{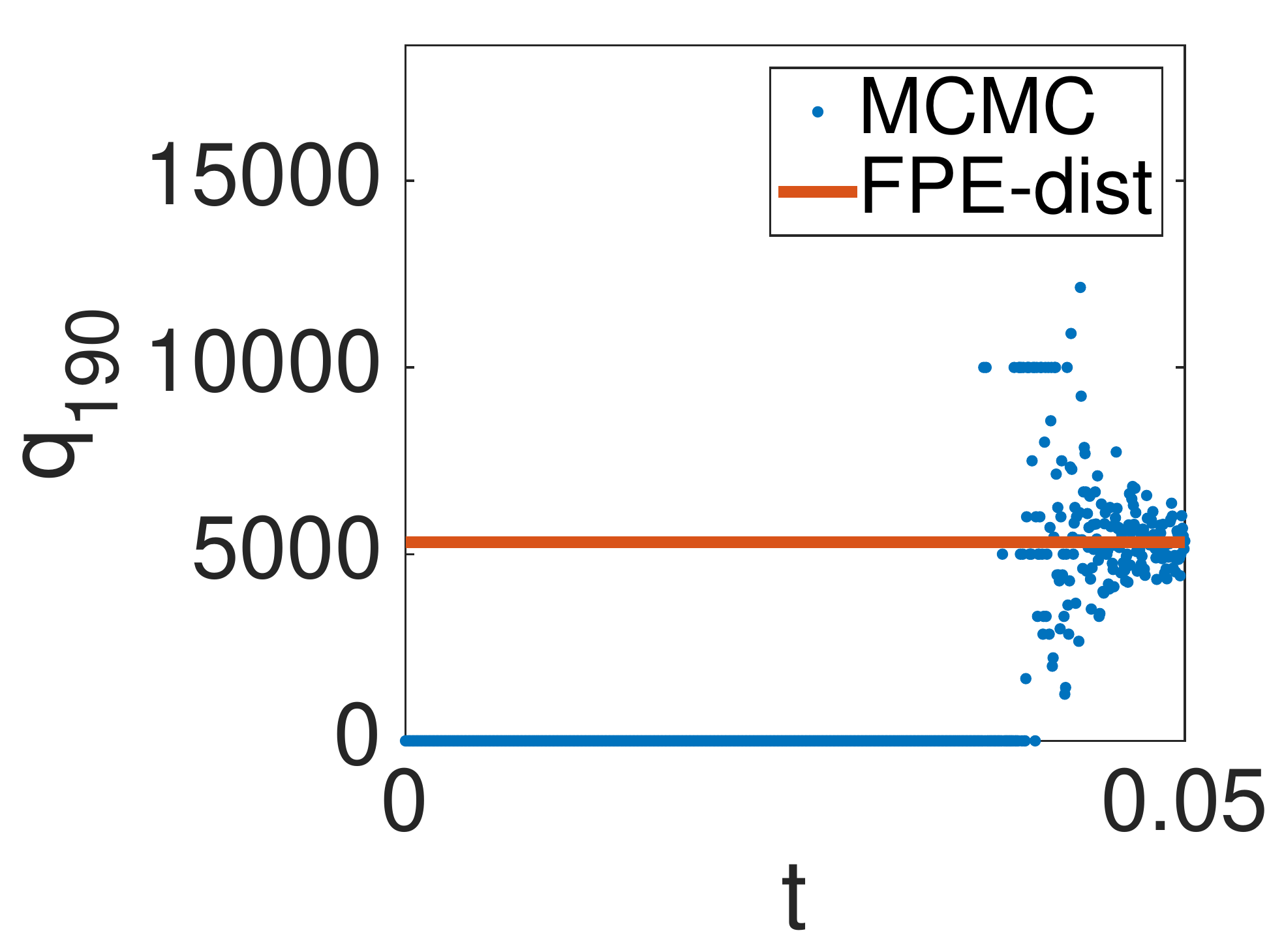}
\includegraphics[width=.24\textwidth]{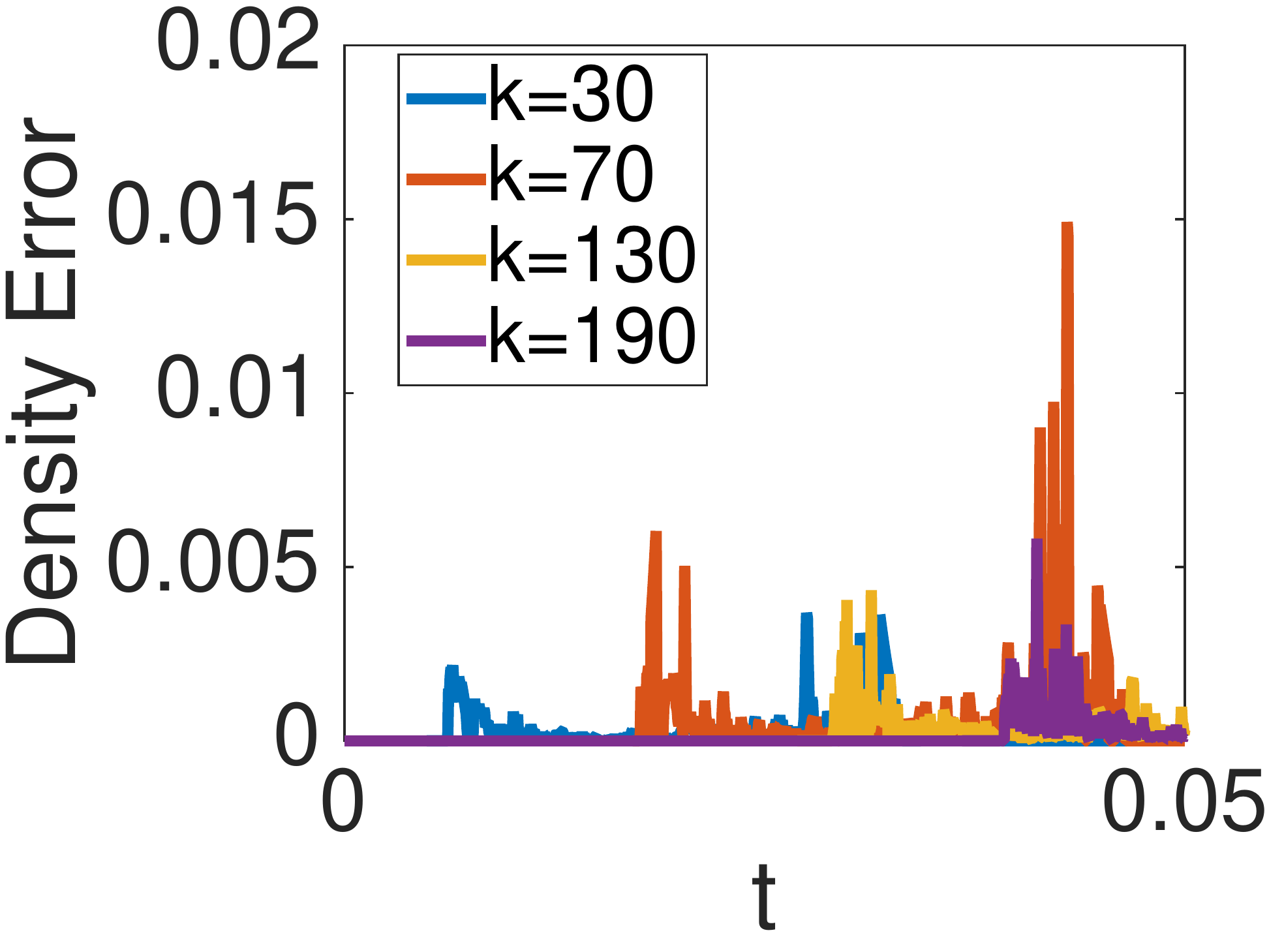}
\includegraphics[width=.24\textwidth]{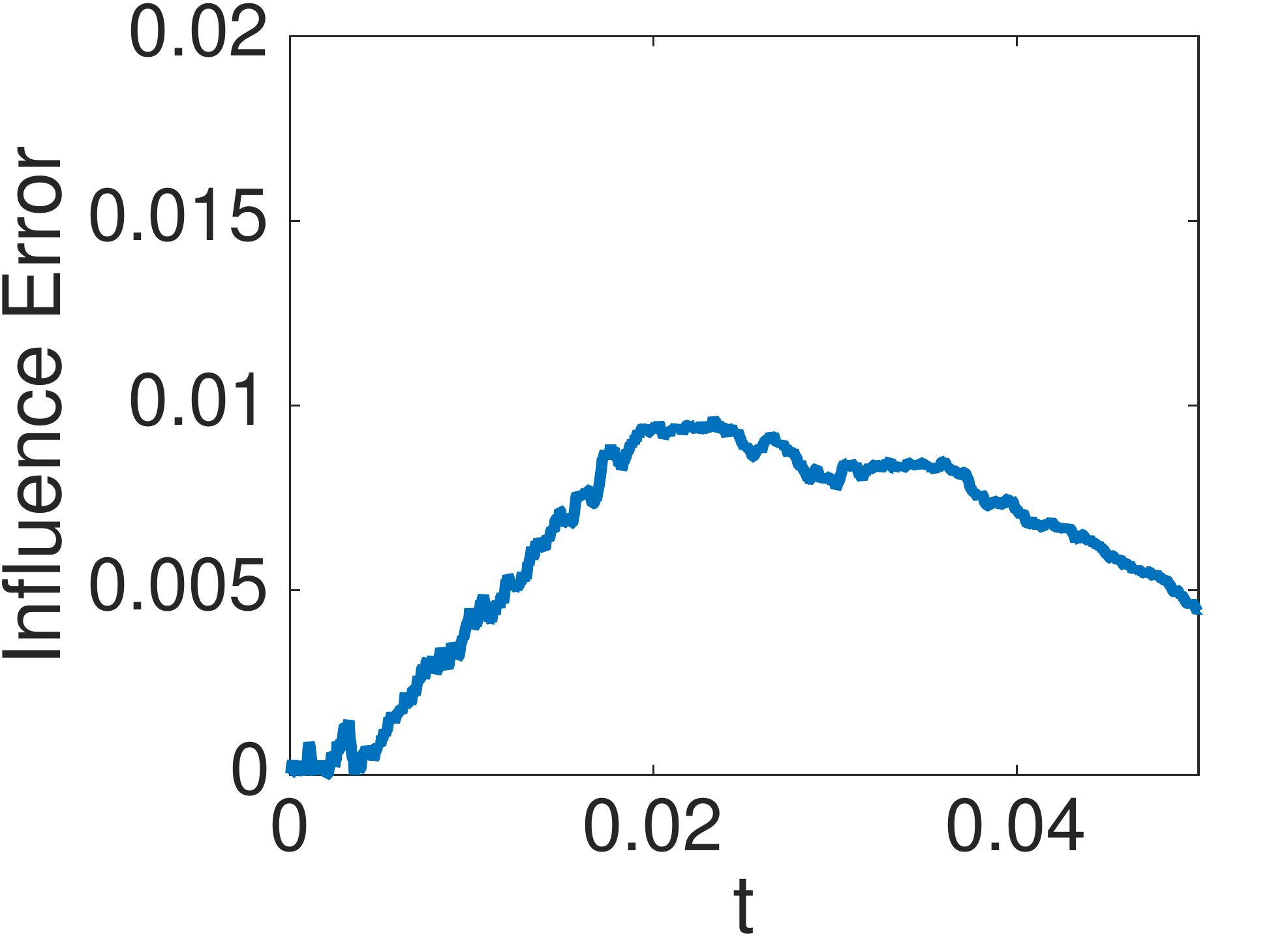}
\includegraphics[width=.24\textwidth]{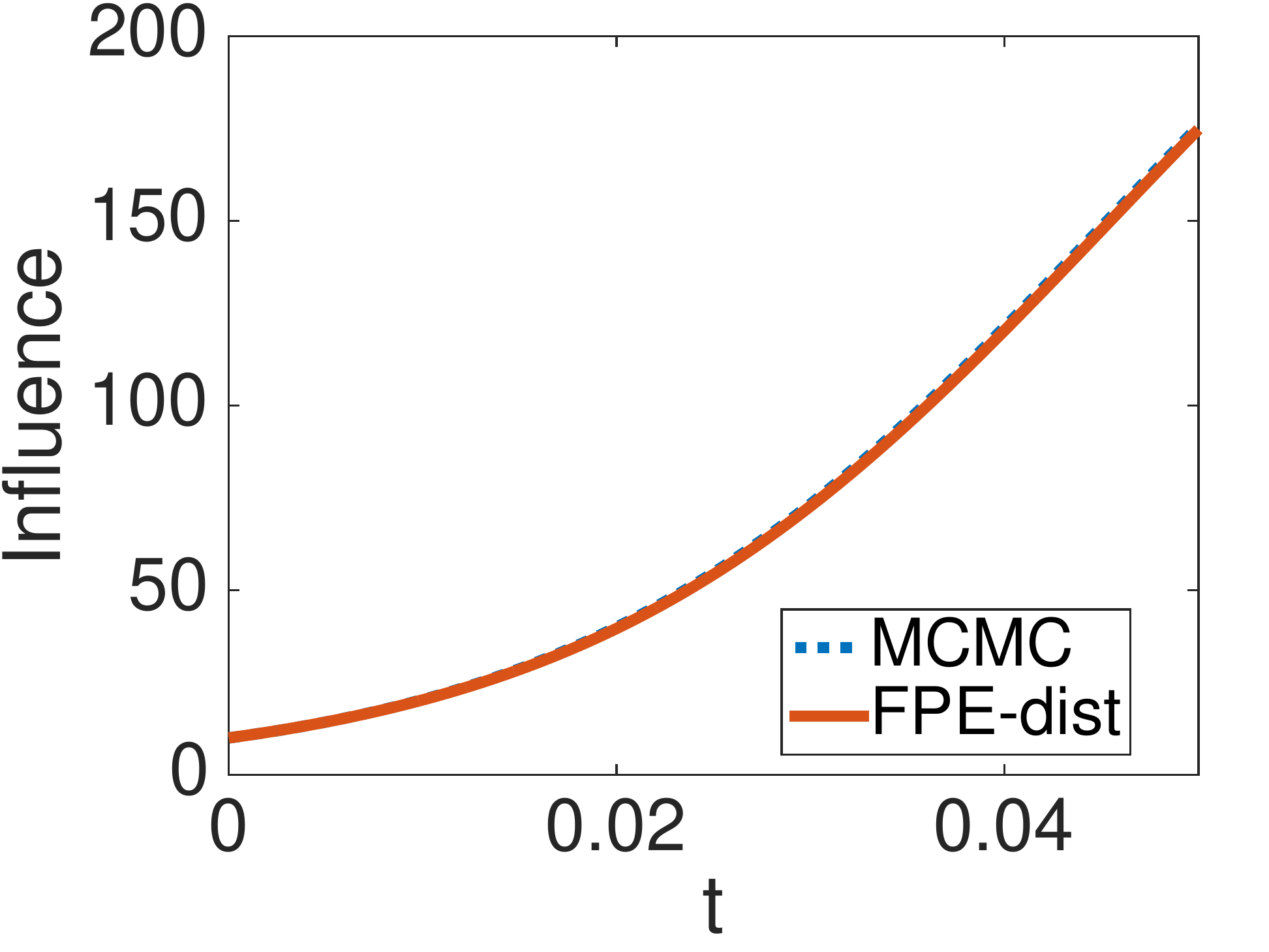}
\includegraphics[width=.24\textwidth]{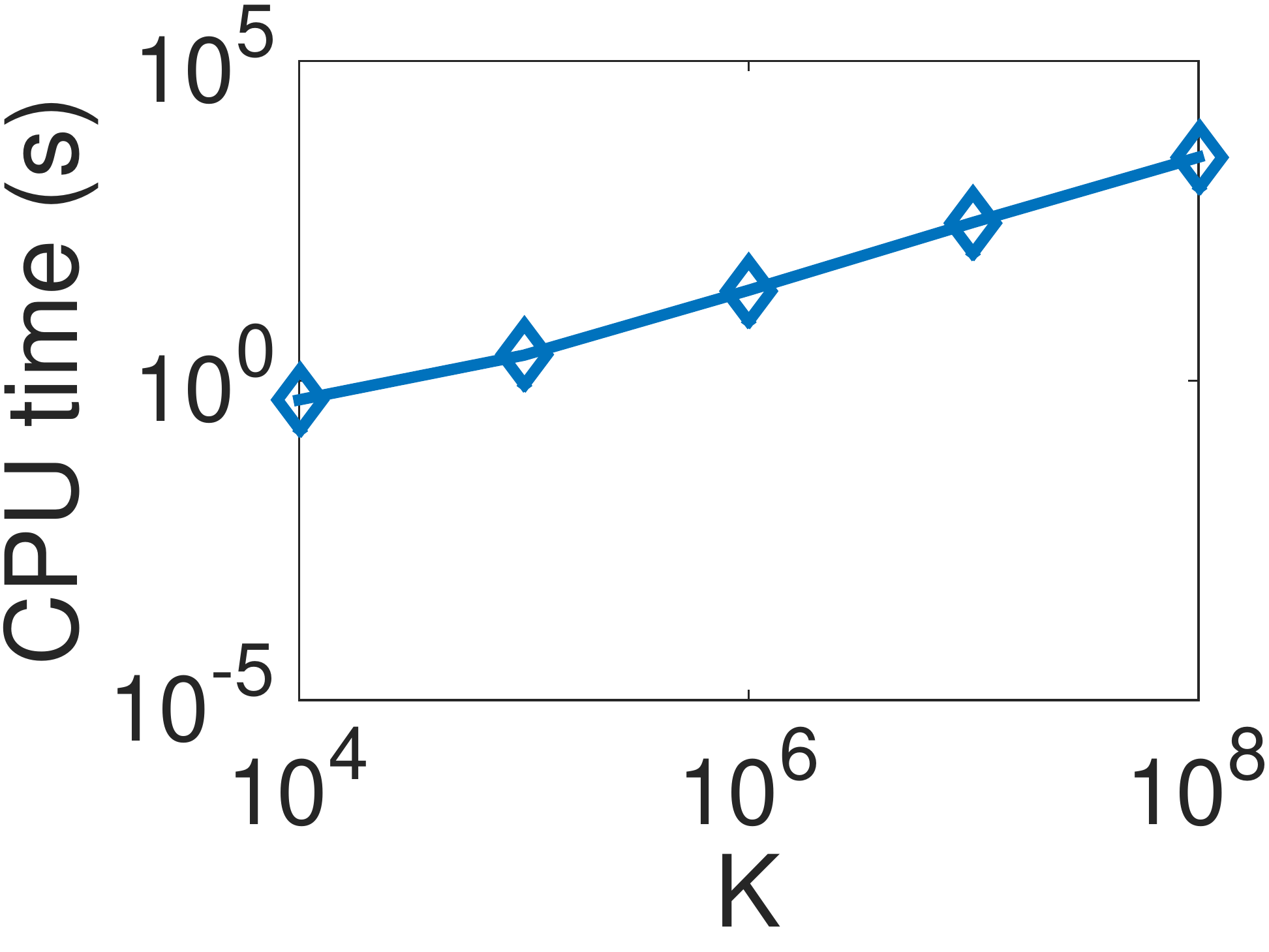}
\caption{\textbf{Top row}: $\qhat_k$ estimated in \texttt{FPE-dist}
and $q_k(t)$ shown by ground truth (MCMC simulation) for $k=10,70,130,190$
using a dense Erd\H{o}s-R\'{e}nyi's network of size $K=300$ and average degree $150$.
\textbf{Bottom row} from left to right: $|\rhohat_k(t)-\rho_k(t)|/\rho_k(t)$;
$|\influhat(t)-\influ(t)|/\influ(t)$ and plot of $\influ(t)$ and $\influhat(t)$ for this $K=300$ network; and 
CPU time (in seconds) of \texttt{FPE-dist} using Runge-Kutta 4th order ODE solver
on networks with $K$ range from $10^4$ to $10^8$.}
\label{fig:cputime}
\end{figure}

To show the great potential of the proposed method for influence prediction on 
large sized networks, we plot the CPU time (in seconds) for solving the
Fokker-Planck equation \eqref{eq:FPEform} numerically using MATLAB with single core
computation on a regular desktop computer (Intel Core 3.4GHz CPU)
in the bottom rightmost panel of Fig.\,\ref{fig:cputime}.
In contrast, most state-of-the-arts learning-based approaches
suffer drastic increase of computational cost for larger or denser networks
due to the significantly amplified number of simulations required to achieve acceptable level of accuracy \cite{Du:2013a}.
On the other hand, the proposed method possesses low computation
complexity and is scalable for large and dense networks. 

\section{Concluding remarks}
We consider the important influence (expected number of activated nodes) 
prediction problem on general heterogeneous
networks. The problem is significantly different from those in classical mathematical
epidemics theory where individual contact network is not considered nor
those in statistical physics where networks are statistically homogeneous 
and nodes are not exactly distinguishable. 
In our problem, the influence depends on the following factors which all 
play critical roles in computations:
the structure of network (directed graph) $G=(V,E)$, 
the activation rates $\{\alpha_{ij}\}$ between every pair of nodes $i$ and $j$
(and self-activation rates $\{\beta_i\}$ and recovery rates $\{\gamma_i\}$
if applicable), and the source set $S$. In this paper, we proposed a novel approach
by calculating the probability $\rho_k(t)$ ($k$ nodes are activated
at time $t$) for all influence sizes $k$ to obtain influence $\influ(t)=\sum_k k\rho_k(t)$. 
To this end, we establish 
the Fokker-Planck equation as a system of deterministic differential equations
that governs the dynamical evolution of $\{\rho_k(t)\}$.
We provide a few instances for estimating the coefficients in the Fokker-Planck equations,
and establish the relation between coefficient estimation error and the final influence
prediction error, which apply to all types of propagation models on general networks.
We conducted a number of numerical experiments which justify the very promising performance
of the proposed approach in terms of accuracy, efficiency and robustness.

Our novel approach also gives rise to a number of new research problems. For example:
How to approximate the transition rates $q_k$ and $r_k$ accurately for general
propagation models (e.g., activation time is not exponentially distributed and
hence the propagation is not Markov)? How to apply the Fokker-Planck equation approach
to influence prediction when only propagation cascade data is available
(i.e., only the activation times and identities are observed during a number of 
propagations but not
the actual network $G=(V,E)$ and/or activation parameters in practice)?
These problems are important from both of theoretical and practical 
points of view, and we plan to investigate them in our future research.

\section*{Appendix}
\begin{proposition}\label{prop:minT}
Let $T_1,T_2,\dots,T_n$ be independent random variables and 
$T_j \sim \exp(\alpha_i)$ for all $j$, then the probability that
$T_i=\min_{1\leq j\leq n}T_j$ is
$\alpha_i/(\sum_{j=1}^{n}\alpha_{j})$,
and the minimum $\min_{1\leq j\leq n}T_j \sim \exp(\sum_{i=1}^n\alpha_i)$.
\end{proposition}
\begin{proof}
The proof is by direct computation and hence details are omitted here.
\end{proof}

\begin{proposition}\label{prop:mixrate}
Let $T_i \sim \exp(\alpha_i)$
and $Y$ be a multinomial random variable such that
$\pr(Y=i) = p_i$ for $i=1,\dots,n$, then the probability density function
of $T_Y$ is $f_{T_Y}(t) = \sum_{i=1}^n p_i \alpha_i e^{-\alpha_it}$,
and the instantaneous hazard rate of point process associated to time $T_Y$ is
$\alpha_{T_Y}(t)=(\sum_{i=1}^n p_i \alpha_i e^{-\alpha_it})/(\sum_{i=1}^n p_i e^{-\alpha_it})$.
In particular, $\alpha_{T_Y}(0)=\sum_{i=1}^{n}p_i\alpha_i$.
\end{proposition}
\begin{proof}
We use the rule of total probability to obtain
\begin{equation}
\pr (T_Y\geq t)=\sum_{i=1}^n \pr (T_Y\geq t | Y=i) \pr(Y=i)
=\sum_{i=1}^n p_ie^{-\alpha_i t}.
\end{equation}
Hence the cumulative distribution function of $T_Y$ is
$F_{T_Y}(t)=1-\pr(T_Y\geq t)$ and 
probability density function is 
$f_{T_Y}(t)=F_{T_Y}'(t)=\sum_{i=1}^n p_i\alpha_ie^{-\alpha_i t}$.
The instantaneous hazard rate is then given by
$\alpha_{T_Y}(t)=f_{T_Y}(t)/\pr(T_Y\geq t)$.
\end{proof}

\bibliographystyle{abbrv}
\bibliography{/Users/xye/Dropbox/Documents/Library/library}

\end{document}